\documentclass[a4paper]{article}
\usepackage{mathptmx}
\usepackage[utf8]{inputenc}
\usepackage[T1]{fontenc}
\usepackage{amsmath}
\usepackage{amsthm}
\usepackage{amsfonts}
\usepackage{mathtools}
\usepackage{authblk}
\usepackage{hyperref}
\usepackage{float}
\usepackage{booktabs}
\usepackage{subcaption}
\let\oldtabular\tabular 
\renewcommand{\tabular}{\normalsize\oldtabular}

\DeclareMathOperator{\sgn}{sgn}

\newcommand{\allzeros}{\ensuremath{\mathbf{0}}}

\newcommand{\abs}[1]{\ensuremath{\mathopen\lvert #1 \mathclose\rvert}}

\newcommand{\Abs}[1]{\ensuremath{\left| #1 \right|}}

\newcommand{\Ball}{\ensuremath{\mathrm{B}}}
\newcommand{\Simplex}{\ensuremath{\Delta}}
\newcommand{\CornerSimplex}{\ensuremath{\Delta_\mathrm{c}}}
\newcommand{\eps}{\ensuremath{\varepsilon}}
\newcommand{\RR}{\ensuremath{\mathbb{R}}}
\newcommand{\QQ}{\ensuremath{\mathbb{Q}}}

\newcommand{\NP}{\ensuremath{\mathrm{NP}}}

\newcommand{\PSPACE}{\ensuremath{\mathrm{PSPACE}}}
\newcommand{\EXPSPACE}{\ensuremath{\mathrm{EXPSPACE}}}

\newcommand{\cETR}{\ensuremath{\exists\RR}}
\newcommand{\cETQ}{\ensuremath{\exists\QQ}}

\newcommand{\PPAD}{\ensuremath{\mathrm{PPAD}}}
\newcommand{\FIXP}{\ensuremath{\mathrm{FIXP}}}

\newcommand{\Exp}{\operatorname*{E}}

\newtheorem{theorem}{Theorem}
\newtheorem{proposition}{Proposition}
\newtheorem{lemma}{Lemma}

\newtheorem{definition}{Definition}
\newtheorem{remark}{Remark}

\newcommand{\theory}{\ensuremath{\mathrm{Th}}}
\newcommand{\Etheory}{\ensuremath{\theory_\exists}}
\newcommand{\ETR}{\ensuremath{\Etheory(\RR)}}
\newcommand{\ETQ}{\ensuremath{\Etheory(\QQ)}}

\newcommand{\QUAD}{\textsc{Quad}}

\newcommand{\NewExistsCC}[1]{
  \expandafter\newcommand\csname Exists#1\endcsname{\ensuremath{\exists\textsc{#1}}}}

\newcommand{\allbot}{\ensuremath{(\bot,\bot,\bot)}}
\newcommand{\allG}{\ensuremath{(G,G,G)}}

\NewExistsCC{NEInABall}
\NewExistsCC{SecondNE}
\NewExistsCC{NEWithLargePayoffs}
\NewExistsCC{NEWithSmallPayoffs}
\NewExistsCC{NEWithLargeTotalPayoff}
\NewExistsCC{NEWithSmallTotalPayoff}
\NewExistsCC{NEWithLargeSupports}
\NewExistsCC{NEWithSmallSupports}
\NewExistsCC{NEWithRestrictingSupports}
\NewExistsCC{NEWithRestrictedSupports}
\NewExistsCC{NonParetoOptimalNE}
\NewExistsCC{NonStrongNE}
\NewExistsCC{ParetoOptimalNE}
\NewExistsCC{StrongNE}
\NewExistsCC{IrrationalNE}
\NewExistsCC{RationalNE}

\NewExistsCC{SNEInABall}
\NewExistsCC{SecondSNE}
\NewExistsCC{SNEWithLargePayoffs}
\NewExistsCC{SNEWithSmallPayoffs}
\NewExistsCC{SNEWithLargeTotalPayoff}
\NewExistsCC{SNEWithSmallTotalPayoff}
\NewExistsCC{SNEWithLargeSupports}
\NewExistsCC{SNEWithSmallSupports}
\NewExistsCC{SNEWithRestrictingSupports}
\NewExistsCC{SNEWithRestrictedSupports}
\NewExistsCC{NonParetoOptimalSNE}
\NewExistsCC{NonStrongSNE}
\NewExistsCC{ParetoOptimalSNE}
\NewExistsCC{StrongSNE}
\NewExistsCC{IrrationalSNE}
\NewExistsCC{RationalSNE}

\NewExistsCC{NonSymmetricNE}

\newcommand{\calS}{\ensuremath{\mathcal{S}}}
\newcommand{\calG}{\ensuremath{\mathcal{G}}}
\newcommand{\calH}{\ensuremath{\mathcal{H}}}
\newcommand{\calD}{\ensuremath{\mathcal{D}}}

\newcommand{\support}{\ensuremath{\mathrm{Supp}}}

\date{January 15, 2020}

\title{On the Computational Complexity of Decision Problems about
  Multi-Player Nash Equilibria\thanks{This paper forms an extension of parts of the master's thesis of the first author and has appeared previously in a preliminary form~\cite{SAGT:BerthelsenH19}. The second author is supported by the Independent Research Fund Denmark under grant no. 9040-00433B.}}

\author{Marie Louisa Tølbøll Berthelsen}
\author{Kristoffer Arnsfelt Hansen}
\affil{Aarhus University}
\begin{document}
\maketitle
\begin{abstract}
  We study the computational complexity of decision problems about
  Nash equilibria in $m$-player games. Several such problems have
  recently been shown to be computationally equivalent to the decision
  problem for the existential theory of the reals, or stated in terms
  of complexity classes, $\cETR$-complete, when $m\geq 3$. We show
  that, unless they turn into trivial problems, they are $\cETR$-hard
  even for 3-player \emph{zero-sum} games.

  We also obtain new results about several other decision problems.
  We show that when $m\geq 3$ the problems of deciding if a game has a
  Pareto optimal Nash equilibrium or deciding if a game has a strong
  Nash equilibrium are $\cETR$-complete. The latter result rectifies a
  previous claim of $\NP$-completeness in the literature. We show that
  deciding if a game has an irrational valued Nash equilibrium is
  $\cETR$-hard, answering a question of Bilò and Mavronicolas, and
  address also the computational complexity of deciding if a game has
  a rational valued Nash equilibrium. These results also hold for
  3-player zero-sum games.

  \begin{sloppypar}
  Our proof methodology applies to corresponding decision problems
  about symmetric Nash equilibria in symmetric games as well, and in
  particular our new results carry over to the symmetric
  setting. Finally we show that deciding whether a symmetric
  $m$-player games has a \emph{non-symmetric} Nash equilibrium is
  $\cETR$-complete when $m\geq 3$, answering a question of Garg,
  Mehta, Vazirani, and Yazdanbod.
\end{sloppypar}
\end{abstract}

\section{Introduction}
Given a finite strategic form $m$-player game the most basic
algorithmic problem is to compute a Nash equilibrium, shown always to
exist by Nash~\cite{AM:Nash51}. The computational complexity of this
problem was characterized in seminal work by Daskalakis, Goldberg, and
Papadmitriou~\cite{SICOMP:DaskalakisGP09} and Chen and
Deng~\cite{FOCS:ChenDeng06} as $\PPAD$-complete for 2-player games and
by Etessami and Yannakakis~\cite{SICOMP:EtessamiY10} as
$\FIXP$-complete for $m$-player games, when $m \geq 3$. Any 2-player
game may be viewed as a 3-player \emph{zero-sum} game by adding a
dummy player, thereby making the class of 3-player zero-sum games a
natural class of games intermediate between 2-player and 3-player
games. The problem of computing a Nash equilibrium for a 3-player
zero-sum game is clearly $\PPAD$-hard and belongs to $\FIXP$, but its
precise complexity appears to be unknown.

Rather than settling for \emph{any} Nash equilibrium, one might be
interested in a Nash equilibrium that satisfies a given property,
e.g.\ giving each player at least a certain payoff. Such a Nash
equilibrium might of course not exist and therefore results in the
basic computational problem of deciding existence. In the setting of
2-player games, the computational complexity of several such problems
was proved to be $\NP$-complete by Gilboa and
Zemel~\cite{GEB:GilboaZemel89}. Conitzer and
Sandholm~\cite{GEB:ConitzerS08} revisited these problems and showed
them, together with additional problems, to be $\NP$-complete even for
symmetric games.

Only recently was the computational complexity of analogous problems
in $m$-player games determined, for $m\geq 3$. Schaefer and
Štefankovič~\cite{TOCS:SchaeferS15} obtained the first such result by
proving $\cETR$-completeness of deciding existence of a Nash
equilibrium in which no action is played with probability larger
than~$\tfrac{1}{2}$ by any player. Garg, Mehta, Vazirani, and
Yazdanbod~\cite{TEAC:GargMVY18} used this to also show
$\cETR$-completeness for deciding if a game has more than one Nash
equilibrium, whether each player can ensure a given payoff in a Nash
equilibrium, and for the two problems of deciding whether the support
sets of the mixed strategies of a Nash equilibrium can belong to given
sets or contain given sets. In addition, by a symmetrization
construction, they show that the analogue to the latter two problems
for symmetric Nash equilibria are $\cETR$-complete as well. Bilò and
Mavronicolas~\cite{STACS:BiloM16,STACS:BiloM17} subsequently extended
the results of Garg~et~al.\ to further problems both about Nash
equilibria and about symmetric Nash equilibria. They show
$\cETR$-completeness of deciding existence of a Nash equilibrium where
all players receive at most a given payoff, where the total payoff of
the players is at least or at most a given amount, whether the size of
the supports of the mixed strategies all have a certain minimum or
maximum size, and finally whether a Nash equilibrium exists that is
\emph{not} Pareto optimal or that is \emph{not} a strong Nash
equilibrium. All the analogous problems about symmetric Nash
equilibria are shown to be $\cETR$-complete as well.

\subsection{Our Results}
We revisit the problems about existence of Nash equilibria in
$m$-player games, with $m\geq 3$, considered by Garg~et~al.\ and Bilò
and Mavronicolas. In a zero-sum game the total payoff of the players
in any Nash equilibrium is of course~0, and any Nash equilibrium is
Pareto optimal. This renders the corresponding decision problems
trivial in the case of zero-sum games. We show except for these, all
the problems considered by Garg~et~al.\ and Bilò and Mavronicolas
remain $\cETR$-hard for 3-player zero-sum games. We obtain our results
building on a recent more direct and simple proof of $\cETR$-hardness
of the initial $\cETR$-complete problem of Schaefer and Štefankovič
due to Hansen~\cite{TCS:Hansen19}. For completeness we give also
comparably simpler proofs of $\cETR$-hardness for the problems about
total payoff and existence of a non Pareto optimal Nash
equilibrium.

We next show that deciding existence of a strong Nash equilibrium in
an $m$-player game with $m\geq 3$ is $\cETR$-complete, and likewise
for the similar problem of deciding existence of a Pareto optimal Nash
equilibrium. Gatti, Rocco, and Sandholm~\cite{AAMAS:GattiRS13} proved
earlier that deciding if a given (rational valued) strategy profile
$x$ is a strong Nash equilibrium can be done in polynomial time. They
then erroneously concluded that the problem of deciding existence of a
strong Nash equilibrium is, as a consequence $\NP$-complete. A problem
with this reasoning is that if a strong Nash equilibrium exists, there
is no guarantee that a rational valued strong Nash equilibrium
exists. Even if one disregards a concern about irrational valued
strong Nash equilibria, it is possible that even when a rational
valued strong Nash equilibrium exists, any rational valued strong Nash
equilibrium could require \emph{exponentially} many bits to describe
in standard binary notation the numerators and denominators of the
probabilities of the equilibrium strategy profile. Nevertheless, our
proof of $\cETR$-membership build on the idea behind the polynomial
time algorithm of Gatti~et~al. Our reduction for proving
$\cETR$-hardness produces non-zero-sum games.  The case of deciding
existence of a Pareto optimal Nash equilibrium is, as already noted,
trivial for the case of 3-player zero-sum games. We leave the
complexity of the deciding existence of a strong Nash equilibrium in
3-player zero-sum games an open problem.

In another work, Bilò and Mavronicolas~\cite{TOCS:BiloM14} considered
the problems of deciding whether an irrational valued Nash equilibrium
exists and whether a rational valued Nash equilibrium exists, proving
both problems to be $\NP$-hard.  Bilò and Mavronicolas asked if the
problem about existence of an irrational valued Nash equilibria is 
hard for the so-called square-root-sum problem. We confirm this,
showing the problem to be $\cETR$-hard. We relate the problem about
existence of rational valued Nash equilibria to the existential theory
of the rationals.

We next use a symmetrization construction similar to Garg~et~al.\ to
translate all problems considered to the analogous setting of decision
problems about symmetric Nash equilibria. Here we do not obtain
qualitative improvements on existing results, but give for
completeness the simple proofs of these results in addition to our new
results.

A final problem we consider is of deciding existence of a
\emph{nonsymmetric} Nash equilibrium of a given symmetric game. Mehta,
Vazirani, and Yazdanbod~\cite{SAGT:MehtaVY15} proved that this problem
is $\NP$-complete for 2-player games, and
Garg~et~al.~\cite{TEAC:GargMVY18} raised the question of the
complexity for $m$-player games with $m\geq 3$. We show this problem to
be $\cETR$-complete.

Our hardness proofs are presented for the special case of 3-player
games, but extend to $m$-player games for any fixed $m > 3$, in a
similar way to previous
works~\cite{TEAC:GargMVY18,STACS:BiloM16,STACS:BiloM17}. For the case
of nonsymmetric games this is achieved by adding $m-3$ dummy players
with suitably chosen actions sets and payoff functions
(cf.~\cite{STACS:BiloM16}). Zero-sum games are, of course, mainly
interesting for 3-player games.  For the case of symmetric games, the
$m-3$ dummy players can be introduced prior to the symmetrization
construction and this together with the reductions that follow are
easily generalized to $m$ players.

\section{Preliminaries}

\subsection{Existential Theory of the Reals and Rationals}

The existential theory of the reals $\ETR$ is the set of all true
sentences over $\RR$ of the form
$\exists x_1,\dots,x_n \in \RR : \phi(x_1,\dots,x_n)$, where $\phi$ is
a quantifier free Boolean formula of equalities and inequalities of
polynomials with integer coefficients. The complexity class $\cETR$ is
defined~\cite{TOCS:SchaeferS15} as the closure of $\ETR$ under
polynomial time many-one reductions. Equivalently, $\cETR$ is the
constant-free Boolean part of the class $\NP_\RR$~\cite{FOCM:BurgisserC09}, which is the
analogue class to $\NP$ in the Blum-Shub-Smale model of
computation~\cite{BAMS:BlumSS89}. It is straightforward to see that
$\ETR$ is \NP-hard (cf.~\cite{JCSS:BussFS99}) and the decision
procedure by Canny~\cite{STOC:Canny88} shows that $\ETR$ belongs to
$\PSPACE$. Thus it follows that
$\NP \subseteq \cETR \subseteq \PSPACE$.

We may similarly consider the existential theory  over the
rationals $\ETQ$ and likewise form the complexity class $\cETQ$ as the
closure of $\ETQ$ under polynomial time many-one reductions. While it
is a long-standing open problem whether $\ETQ$ is decidable,
Koenigsmann~\cite{AnnMath:Koenigsmann16} recently showed that already
$\theory_{\forall\exists}(\QQ)$, consisting of true sentences in
prenex form with a single block of universal quantifiers followed by a
single block of existential quantifiers, is undecidable. In contrast,
the entire first order theory $\theory(\RR)$ of the reals is decidable
in \EXPSPACE~\cite{JSC:Renegar92}. Schaefer and
Štefankovič~\cite{TOCS:SchaeferS15} show that the problem of deciding
feasibility of a system of \emph{strict} inequalities is complete for
$\cETR$. Since a system of strict inequalities that is feasible over
$\RR$ is also feasible over $\QQ$, it follows that
$\cETR \subseteq \cETQ$.

The basic complete problem for $\cETR$ and for $\cETQ$, is the problem
of deciding whether a system of quadratic equations with integer
coefficients has a solution over $\RR$ and over $\QQ$,
respectively~\cite{BAMS:BlumSS89}. We denote this problem over $\RR$
as $\QUAD$ and the problem over $\QQ$ as $\QUAD_\QQ$.

\subsection{Strategic Form Games and Nash Equilibrium}
A finite strategic form game $\calG$ with $m$ players is given by sets
$S_1,\dots,S_m$ of actions (\emph{pure strategies}) together with
\emph{utility functions}
$u_1,\dots,u_m : S_1 \times \dots \times S_m \rightarrow \RR$. A
choice of an action $a_i\in S_i$ for each player together form a pure
strategy profile $a=(a_1,\dots,a_m)$.

\begin{sloppypar}
The game $\calG$ is \emph{symmetric} if $S_1=\dots=S_m$ and for every
permutation $\pi$ on $[m]$, every $i \in [m]$ and every
$(a_1,\dots,a_m) \in S_1\times\dots\times S_m$ it holds that
$u_i(a_1,\dots,a_m)= u_{\pi(i)}(a_{\pi(1)},\dots,a_{\pi(m)})$. In other
words, a game is symmetric if the players share the same set of
actions and the utility function of a player depends only on the
action of the player together with the \emph{multiset} of actions of
the other players.
\end{sloppypar}

Let $\Delta(S_i)$ denote the set of probability distributions on
$S_i$.  A \emph{(mixed) strategy} for Player~$i$ is an element
$x_i \in \Delta(S_i)$. The \emph{support} $\support(x_i)$ is the set
of actions given strictly positive probability by $x_i$. We say that
$x_i$ is \emph{fully mixed} if $\support(x_i)=S_i$. A strategy $x_i$
for each player $i$ together form a strategy profile
$x=(x_1,\dots,x_m)$. The utility functions  extend to
strategy profiles by letting
$u_i(x)=\Exp_{a \sim x}u_i(a_1,\dots,a_m)$. We shall also refer to
$u_i(x)$ as the \emph{payoff} of Player~$i$.

Given a strategy profile $x$ we let
$x_{-i}=(x_1,\dots,x_{i-1},x_{i+1},\dots,x_m)$ denote the strategies
of all players except Player~$i$. Given a strategy $y \in S_i$ for
Player~$i$, we let $(x_{-i};y)$ denote the strategy profile
$(x_1,\dots,x_{i-1},y,x_{i+1},\dots,x_m)$ formed by $x_{-i}$ and
$y$. We may also denote $(x_{-i};y)$ by $x\setminus y$. We say that
$y$ is a \emph{best reply} for Player~$i$ to $x$ (or to $x_{-i}$) if
$u_i(x\setminus y) \geq u_i(x\setminus y')$ for all
$y'\in\Delta(S_i)$.

A \emph{Nash equilibrium} (NE) is a strategy profile $x$ where each
individual strategy $x_i$ is a best reply to $x$. As shown by
Nash~\cite{AM:Nash51}, every finite strategic form game $\calG$ has a
Nash equilibrium. In a symmetric game $\calG$, a \emph{symmetric Nash
  equilibrium} (SNE) is a Nash equilibrium where the strategies of all
players are identical. Nash also proved that every symmetric game has
a symmetric Nash equilibrium.

A strategy profile $x$ is \emph{Pareto optimal} if there is no
strategy profile $x'$ such that $u_i(x)\leq u_i(x')$ for all $i$, and
$u_j(x) < u_j(x')$ for some $j$. A Nash equilibrium strategy profile
need not be Pareto optimal and a Pareto optimal strategy profile need
not be a Nash equilibrium. A strategy profile that is both a Nash
equilibrium and is Pareto optimal is called a Pareto optimal Nash
equilibrium. The existence of a Pareto optimal Nash equilibrium is not
guaranteed.

A \emph{strong Nash equilibrium}~\cite{PJM:Aumann60} (strong NE) is a
strategy profile $x$ for which there is no non-empty set
$B \subseteq [m]$ for which \emph{all} players $i \in B$ can increase
their payoff by different strategies assuming players
$j \in [m]\setminus B$ play according to $x$. Equivalently, $x$ is a
strong Nash equilibrium if for every strategy profile $x'\neq x$ there
exist $i$ such that $x_i \neq x'_i$ and $u_i(x') \leq u_i(x)$.  The
existence of a strong Nash equilibrium is not guaranteed.

\section{Decision Problems about Nash Equilibria}
\label{SEC:DecisionProblemsNE}

Below we define the decision problems under consideration with names
generally following Bilò and Mavronicolas~\cite{STACS:BiloM16}. The given input is a finite
strategic form game $\calG$, together with auxiliary input depending
on the particular problem. We let $u$ denote a rational number, $k$ an
integer, and $T_i \subseteq S_i$ a set of actions of Player~$i$, for
every $i$. We describe the decision problem by stating the property a
Nash equilibrium $x$ whose existence is to be determined should
satisfy. The problems are grouped together in four groups each of
which are covered in a separate subsection below.

Except for the last four problems, it is straightforward to
prove membership in $\cETR$ by an explicit existentially quantified
first-order formula. We prove $\cETR$ membership of
$\ExistsParetoOptimalNE$ and $\ExistsStrongNE$ in
subsection~\ref{SEC:ParetoStrongNE} and discuss decidability of
$\ExistsIrrationalNE$ and $\ExistsRationalNE$ in
subsection~\ref{SEC:IrrationalRationNE}.
\begin{center}
\begin{tabular}{ll}\toprule
  Problem & Condition\\\midrule
\ExistsNEWithLargePayoffs & $u_i(x)\geq u$ for all~$i$.\\
\ExistsNEWithSmallPayoffs & $u_i(x)\leq u$ for all~$i$.\\
\ExistsNEWithLargeTotalPayoff & $\sum_i u_i(x)\geq u$.\\
\ExistsNEWithSmallTotalPayoff & $\sum_i u_i(x)\leq u$.\\\midrule
\ExistsNEInABall & $x_i(a_i)\leq u$ for all~$i$ and $a_i\in S_i$.\\
\ExistsSecondNE & $x$ is not the only NE.\\
\ExistsNEWithLargeSupports & $\abs{\support(x_i)}\geq k$ for all~$i$.\\
\ExistsNEWithSmallSupports &  $\abs{\support(x_i)}\leq k$ for all~$i$.\\
\ExistsNEWithRestrictingSupports &  $T_i\subseteq \support(x_i)$ for all~$i$.\\
  \ExistsNEWithRestrictedSupports & $\support(x_i) \subseteq T_i$ for all~$i$.\\\midrule
\ExistsNonParetoOptimalNE & $x$ is not Pareto optimal.\\
\ExistsNonStrongNE & $x$ is not a strong NE.\\
\ExistsParetoOptimalNE & $x$ is Pareto optimal.\\
\ExistsStrongNE & $x$ is a strong NE.\\\midrule
\ExistsIrrationalNE & $x_i(a_i) \not\in \QQ$ for some $i$ and $a_i\in S_i$.\\
  \ExistsRationalNE & $x_i(a_i) \in \QQ$ for all $i$ and $a_i \in S_i$.\\\bottomrule  
\end{tabular}
\end{center}

A key step (implicitly present) in the proof of the first
$\cETR$-hardness result about Nash equilibrium in 3-player games by
Schaefer and Štefankovič is a result due to
Schaefer~\cite{GD:Schaefer09} that \QUAD\ remains $\cETR$-hard under
the \emph{promise} that either the given quadratic system has no
solutions or a solution exists in the unit ball
$\Ball(\allzeros,1)$. For our purposes the following variation
\cite[Proposition~2]{TCS:Hansen19} will be more directly applicable
(which may easily be proved from the latter,
cf.~Section~\ref{SEC:IrrationalRationNE}). Here we denote by
$\CornerSimplex^n$ the standard corner $n$-simplex
$\{x \in \RR^n \mid x\geq 0 \wedge \sum_{i=1}^n x_i \leq 1\}$.
\begin{proposition}
  \label{PROP:QUAD-CornerSimplex} 
  It is $\cETR$-hard to decide if a given system of quadratic
  equations in~$n$ variables and with integer coefficients has a
  solution under the promise that either the system has no solutions
  or a solution $z$ exists that is in the interior of
  $\CornerSimplex^n$ and also satisfies $z_i\leq \frac{1}{2}$ for all
  $i$ and that $\sum_{i=1}^n z_i \geq \frac{1}{2}$.
\end{proposition}
Schaefer and Štefankovič showed that $\ExistsNEInABall$ is
$\cETR$-hard for 3-player ga\-mes by first proving that the following
problem is $\cETR$-hard: Given a continuous function
$f : \Ball(\allzeros,1) \rightarrow \Ball(\allzeros,1)$ mapping the
unit ball to itself, where each coordinate function $f_i$ is given as
a polynomial, and given a rational number $r$, is there a fixed point
of $f$ in the ball $\Ball(\allzeros,r)$? The proof was then concluded
by a transformation of Brouwer functions into 3-player games by
Etesammi and Yannakakis~\cite{SICOMP:EtessamiY10}. This latter
reduction is rather involved and goes though an intermediate
construction of 10-player games. More recently,
Hansen~\cite{TCS:Hansen19} gave a simple and direct reduction from the
above promise version of $\QUAD$ to $\ExistsNEInABall$.

The first step of this as well as our reductions is to transform 
the given quadratic system over the corner simplex $\CornerSimplex^n$
into a homogeneous bilinear system over the standard $n$-simplex
$\{x \in \RR^{n+1} \mid x\geq 0 \wedge \sum_{i=1}^{n+1} x_i = 1\}$
which we denote by $\Simplex^n$. In short, this is done by introducing
a set of new variables $y_i$ and new equations $x_i-y_i=0$, replacing
quadratic terms $x_ix_j$ by bilinear quadratic terms $x_iy_j$, and
finally homogenizing the entire system using the two equations
$\sum_{i=1}^{n+1} x_i=1$ and $\sum_{i=1}^{n+1} y_i=1$ where $x_{n+1}$ and
$y_{n+1}$ are new \emph{slack variables}. Doing this we arrive at the
following statement (cf.~\cite[Proposition~3]{TCS:Hansen19}).
\begin{proposition}
\label{PROP:BilinearSystem}
It is $\cETR$-complete to decide if a system of homogeneous bilinear
equations $q_k(x,y)=0$, $k=1,\dots,\ell$ with integer coefficients has
a solution $x,y \in \Simplex^n$. It remains $\cETR$-hard under the
promise that either the system has no such solution or a solution
$(x,x)$ exists where $x$ belong to the relative interior of
$\Simplex^n$ and further satisfies $x_i \leq \frac{1}{2}$ for all $i$.
\end{proposition}
\subsection{Payoff Restricted Nash Equilibria}
\label{sec:payoff-restr-nash}
For proving the $\cETR$-hardness results we start by showing that it
is $\cETR$-hard to decide if a given zero-sum game has a Nash
equilibrium in which each player receives payoff~$0$. This is in
contrast to the earlier work of Garg~et~al.~\cite{TEAC:GargMVY18} and
Bilò and Mavronicolas~\cite{STACS:BiloM16,STACS:BiloM17} that reduce
from the $\ExistsNEInABall$ problem. On the other hand we do show
$\cETR$-hardness even under the promise that the Nash equilibrium also
satisfies the condition of $\ExistsNEInABall$. The construction and
proof below are modifications of proofs by Hansen~\cite[Theorem~1 and
Theorem~2]{TCS:Hansen19}.

\begin{definition}[The 3-player zero-sum game $\calG_0$]
  \begin{sloppypar}
   Let $\calS$ be a system of homogeneous bilinear polynomials
  $q_1(x,y),\dots,q_\ell(x,y)$ with integer coefficients in variables
  $x=(x_1,\dots,x_{n+1})$ and $y=(y_1,\dots,y_{n+1})$,
  \[
    q_k(x,y)=\sum_{i=1}^{n+1}\sum_{j=1}^{n+1}a_{ij}^{(k)}x_iy_j \enspace .
  \]
\end{sloppypar}
We define the 3-player game $\calG_0(\calS)$ as follows. The
  strategy set of Player~1 is the set
  $S_1=\{1,-1\} \times \{1,2,\dots,\ell\}$.  The strategy sets of
  Player~2 and Player~3 are $S_2=S_3=\{1,2,\dots,n+1\}$.  The
  (integer) utility functions of the players are defined by
  \[
    \tfrac{1}{2}u_1((s,k),i,j) = -u_2((s,k),i,j) = -u_3((s,k),i,j) =
    sa_{ij}^{(k)} \enspace .
  \]
  \label{DEF:G0}
\end{definition}
When the system $\calS$ is understood by the context, we simply write
$\calG_0=\calG_0(\calS)$. We think of the strategy $(s,k)$ of Player~1
as corresponding to the polynomial $q_k$ together with a sign $s$, the
strategy $i$ of Player~2 as corresponding to $x_i$ and the strategy
$j$ of Player~3 as corresponding to $y_j$. We may thus identify mixed
strategies of Player~2 and Player~3 with assignments to variables
$x,y \in \Simplex^n \subseteq \RR^{n+1}$.

The following observation is immediate from the definition of
$\calG_0$.
\begin{lemma}
  Any strategy profile $(x,y)$ of Player~2 and Player~3 satisfies for
  every $(s,k)\in S_1$ the equation
\begin{equation}
  \tfrac{1}{2}u_1((s,k),x,y)= -u_2((s,k),x,y) = -u_3((s,k),x,y) =sq_k(x,y) \enspace .
  \label{EQ:G0payoff}
\end{equation}
Hence $u_1(z,x,y)=u_2(z,x,y)=u_3(z,x,y)=0$ when $z$ is the uniform
distribution on $S_1$. Consequentially, any Nash equilibrium payoff
profile is of the form $(2u,-u,-u)$, where $u\geq 0$.
\label{LEM:G0payoff}
\end{lemma}
Next we relate solutions to the system $\calS$ to Nash equilibria in
$\calG_0$.
\begin{proposition}
  Let $\calS$ be a system of homogeneous bilinear polynomials
  $q_k(x,y)$, $k=1,\dots,\ell$. If $\calS$ has a solution
  $(x,y) \in \Simplex^n\times\Simplex^n$, then letting $z$ be the
  uniform distribution on $S_1$, the strategy profile $\sigma=(z,x,y)$
  is a Nash equilibrium of $\calG_0$ in which every player receives
  payoff~$0$. If in addition $(x,y)$ satisfies the promise of
  Proposition~\ref{PROP:BilinearSystem}, then $\sigma$ is fully mixed,
  Player~2 and Player~3 use identical strategies, and no action is
  chosen with probability more than~$\frac{1}{2}$ by any player.
  Conversely, if $(z,x,y)$ is a Nash equilibrium of $\calG_0$ in which
  every player receives payoff~$0$, then $(x,y)$ is a solution to
  $\calS$.
\label{PROP:G0NE}
\end{proposition}
\begin{proof}
  Suppose first that $(x,y) \in \Simplex^n\times\Simplex^n$ is a
  solution to $\calS$ and let $z$ be the uniform distribution on
  $S_1$. By Equation~\eqref{EQ:G0payoff} the strategy profile $(x,y)$
  of Player~2 and Player~3 ensures that all players receive payoff~$0$
  regardless of which strategy is played by Player~1, and likewise the
  strategy $z$ of Player~1 ensures that all players receive payoff~$0$
  regardless of the strategies of Player~2 and Player~3. This shows
  that $\sigma$ is a Nash equilibrium of $\calG_0$, in which by
  Lemma~\ref{LEM:G0payoff} every player receives payoff~$0$. If $(x,y)$
  in addition satisfies that the promise of
  Proposition~\ref{PROP:BilinearSystem} we have
  $0<x_i=y_i\leq \frac{1}{2}$. From this and our choice of $z$, we
  have that $\sigma$ is a fully mixed and that no action is chosen by
  a strategy of $\sigma$ with probability more than~$\frac{1}{2}$.

  Suppose on the other hand that $\sigma=(z,x,y)$ is a Nash
  equilibrium of $\calG_0$ in which every player receives payoff~$0$.
  Suppose that $q_k(x,y) \neq 0$ for some $k$. Then by
  Equation~\eqref{EQ:G0payoff} we get that
  $u_1((\sgn(q_k(x,y)),k),x,y) = \abs{2q_k(x,y)}>0$, contradicting that
  $\sigma$ is a Nash equilibrium. Thus $(x,y)$ is a solution to $\calS$.
\end{proof}

\begin{theorem}
  \label{THM:NEWithLargePayoffs}
  \begin{sloppypar}
$\ExistsNEWithLargePayoffs$ and 
$\ExistsNEWithSmallPayoffs$ are $\cETR$-complete, even for 3-player
zero-sum games.
\end{sloppypar}
\end{theorem}
\begin{proof}
  For a strategy profile $x$ in a zero-sum game $\calG$ we have that
  $u_i(x)=0$, for all~$i$, if and only if $u_i(x)\geq 0$, for all~$i$,
  if and only if $u_i(x)\leq 0$, for all~$i$.

  Thus Proposition~\ref{PROP:G0NE} gives a reduction from the promise
  problem of Proposition~\ref{PROP:BilinearSystem}, thereby
  establishing $\cETR$-hardness of the problems
  $\ExistsNEWithLargePayoffs$ and \linebreak
  $\ExistsNEWithSmallPayoffs$.
\end{proof}
\begin{sloppypar}
  A simple change to the game $\calG_0$ give
$\cETR$-hardness for the two problems
$\ExistsNEWithLargeTotalPayoff$ and
$\ExistsNEWithSmallTotalPayoff$. Naturally  we must give up the
zero-sum property of the game.
\end{sloppypar}
\begin{theorem}[Bilò and Mavronicolas~\cite{STACS:BiloM16}]
  \begin{sloppypar}
    $\ExistsNEWithLargeTotalPayoff$ and
  $\ExistsNEWithSmallTotalPayoff$ are $\cETR$-complete, even for
  3-player games.
\end{sloppypar}
\end{theorem}
\begin{proof}
  Define the game $\calG'_0$ from $\calG_0$ with new utility functions
  $u'_1(x)=u_1(x)$ and $u'_2(x)=u'_3(x)=-u_1(x)$, and thus also
  $u'_2(x)=u'_3(x)=2u_2(x)=2u_3(x)$, where $u_1$,$u_2$, and $u_3$ are
  the utility functions of $\calG_0$. Clearly $\calG'_0$ has the same
  set of Nash equilibria as $\calG_0$. Now
  $u'_1(x)+u'_2(x)+u'_3(x) = -u_1(x)$ and it follows that
  $u'_1(x)+u'_2(x)+u'_3(x)\geq 0$ if and only if $u_1(x)\leq 0$. By
  Lemma~\ref{LEM:G0payoff}, any Nash equilibrium $x$ must satisfy the
  inequality $u_1(x)\geq 0$. Thus, a Nash equilibrium $x$ satisfies
  the inequality $u'_1(x)+u'_2(x)+u'_3(x)\geq 0$ if and only if
  $u_1(x)=u_2(x)=u_3(x)=0$. We conclude that
  Proposition~\ref{PROP:G0NE} gives a reduction from the promise
  problem of Proposition~\ref{PROP:BilinearSystem} to
  $\ExistsNEWithLargeTotalPayoff$ thereby showing $\cETR$-hardness.

  Similarly, define the game $\calG''_0$ from $\calG_0$ with new
  utility functions $u''_1(x)=3u_1(x)$ and
  $u''_2(x)=u''_3(x)=-u_1(x)$. Again, $\calG''_0$ clearly has the same
  set of Nash equilibria as $\calG_0$. Now
  $u''_1(x)+u''_2(x)+u''_3(x) = u_1(x)$ and it follows that
  $u''_1(x)+u''_2(x)+u''_3(x)\leq~0$ if and only if $u_1(x)\leq
  0$. Analogously to above we then obtain $\cETR$-hardness for
  $\ExistsNEWithSmallTotalPayoff$.
\end{proof}
\subsection{Probability Restricted Nash Equilibria}
\label{SEC:ProbabilityRestrictedNE}
A key property of the game $\calG_0$ is that Player~1 may ensure all
players receive payoff~$0$. We now give \emph{all} players this choice
by playing a new additional action~$\bot$. We then design the utility
functions involving~$\bot$ in such a way that the pure strategy
profile $\allbot$ is always a Nash equilibrium, and every other Nash
equilibrium is a Nash equilibrium in $\calG_0$ in which all players
receive payoff~$0$.
\begin{definition}
  For $u\geq0$, let $\calH_1=\calH_1(u)$ be the 3-player zero-sum game where each player has the action set $\{G,\bot\}$ and the payoff vectors are given by the
  entries of the following two matrices, where Player~1 selects the
  matrix, Player~2 selects the row, Player~3 selects the column.
  \begin{table}[H]
  \captionsetup[subtable]{labelformat=empty}
  \begin{subtable}[t]{0.49\textwidth}
  \centering
  \begin{tabular}{c|c|c|}
    \multicolumn{1}{c}{} & \multicolumn{1}{c}{$G$} & \multicolumn{1}{c}{$\bot$}\\\cline{2-3}
    $G$ & $(2u,-u,-u)$ & $(\phantom{-}1,-1,0)$\\\cline{2-3}
    $\bot$ & $(\phantom{2u}\mathllap{1},\phantom{-u}\mathllap{0},\phantom{-u}\mathllap{-1})$ & $(-4,\phantom{-}2,2)$\\\cline{2-3}
  \end{tabular}
  \caption{$G$}
\end{subtable}
\begin{subtable}[t]{0.49\textwidth}
  \centering
  \begin{tabular}{c|c|c|}
    \multicolumn{1}{c}{} & \multicolumn{1}{c}{$G$} & \multicolumn{1}{c}{$\bot$}\\\cline{2-3}
    $G$ & $(0,0,\phantom{-}0)$ & $(\phantom{-}2,-3,1)$\\\cline{2-3}
    $\bot$ & $(2,1,-3)$ & $(-2,\phantom{-}1,1)$\\\cline{2-3}
  \end{tabular}
  \caption{$\bot$}
\end{subtable}
\end{table}
\end{definition}
It is straightforward to determine the Nash equilibria of $\calH_1$.
\begin{lemma}
  When $u>0$, the only Nash equilibrium of $\calH_1(u)$ is the pure
  strategy profile $\allbot$. When $u=0$ the only Nash
  equilibria of $\calH_1(u)$ are the pure strategy profiles $\allG$
  and $\allbot$.
  \label{LEM:H1NE}
\end{lemma}
\begin{proof}
  Let $p_i$ be the probability of Player~$i$ choosing the action
  $G$. Consider first the case of $u>0$. Then the action $\bot$ is
  strictly dominating the action $G$ for both Player~2 and Player~3.
  Hence any Nash equilibrium would require $p_2=p_3=0$. The only best
  reply for Player~1 is then $p_1=0$ as well. Consider next the case
  of $u=0$. In case $p_1<1$, again the action $\bot$ is strictly
  dominating the action $G$ for both Player~2 and Player~3, and we
  conclude that $p_1=p_2=p_3=0$ as before. Suppose now that
  $p_1=1$. In a Nash equilibrium we would have either $p_2=p_3=1$ or
  $p_2=p_3=0$. The former clearly gives a Nash equilibrium whereas for
  $p_2=p_3=0$ the only best reply for Player~1 is $p_1=0$.
\end{proof}

We use the game $\calH_1(u)$ to extend the game $\calG_0$. The action
$G$ of $\calH_1$ represents selecting an action from $\calG_0$, and
the payoff vector $(2u,-u,-u)$ that is the result of all players
playing the action $G$ is precisely of the form of the Nash
equilibrium payoff profile of $\calG_0$.
\begin{definition}[The 3-player zero-sum game $\calG_1$]
  Let $\calG_1=\calG_1(\calS)$ be the game obtained from
  $\calG_0(\calS)$ as follows. Each player is given an additional
  action $\bot$. When no player plays the action $\bot$, the payoffs
  are the same as in $\calG_0$. When at least one player is playing
  the action $\bot$ the payoff are the same as in $\calH_1$, where
  each action different from $\bot$ is translated to action $G$.
  \label{DEF:G1}
\end{definition}
We next characterize the Nash equilibria in $\calG_1$.
\begin{proposition}
  The pure strategy profile $\allbot$ is a Nash equilibrium
  of $\calG_1$. Any other Nash equilibrium $x$ in $\calG_1$ is also a
  Nash equilibrium of $\calG_0$ and is such that every player receives
  payoff~$0$.
  \label{PROP:G1NE}
\end{proposition}
\begin{proof}
  By Lemma~\ref{LEM:G0payoff} a Nash equilibrium of $\calG_1$
  induces a Nash equilibrium of $\calH_1(u)$, where $(2u,-u,-u)$ is a
  Nash equilibrium payoff profile of $\calG_0$, by letting each player
  play the action $G$ with the total probability of which the actions
  of $\calG_0$ are played. By Lemma~\ref{LEM:H1NE}, any Nash
  equilibrium in $\calG_1$ different from $\allbot$ must then be a
  Nash equilibrium of $\calG_0$ with Nash equilibrium payoff profile
  $(0,0,0)$ as claimed.  \end{proof}
\begin{theorem}
  \begin{sloppypar}
  The following problems are $\cETR$-complete, even for 3-player
  zero-sum games: \ExistsNEInABall, \ExistsSecondNE, \ExistsNEWithLargeSupports,
  \ExistsNEWithRestrictingSupports, and \ExistsNEWithRestrictedSupports.
\end{sloppypar}
\label{THM:G1NE}
\end{theorem}
\begin{proof}
  Proposition~\ref{PROP:G0NE} and Proposition~\ref{PROP:G1NE} together
  gives a reduction from the promise problem of
  Proposition~\ref{PROP:BilinearSystem} to all of the problems under
  consideration when setting the additional parameters as follows.
  For $\ExistsNEInABall$ we let $u=\tfrac{1}{2}$, we let $k=2$ for
  $\ExistsNEWithLargeSupports$, and lastly we let $T_i$ be the set of
  all actions of Player~$i$ except~$\bot$ for both of the problems
  $\ExistsNEWithRestrictingSupports$ and
  $\ExistsNEWithRestrictedSupports$. 
\end{proof}
\begin{remark}
  Except for the case of $\ExistsSecondNE$, the results of
  Theorem~\ref{THM:G1NE} can also be proved with the slightly simpler
  construction of adding an additional action $\bot$ to the players
  in $\calG_0$ which when played by at least one player results in all
  players receiving payoff~$0$.
\end{remark}
To adapt the reduction of Theorem~\ref{THM:G1NE} to
$\ExistsNEWithSmallSupports$ we need to replace the trivial Nash
equilibrium $\allbot$ by a Nash equilibrium with large support.
\begin{definition}
  Define the 2-player zero-sum game $\calH_2(k)$ as follows. The two
  players, which we denote Player~2 and Player~3, have the same set of pure strategies
  $S_2=S_3=\{0,1,\dots,k-1\}$. The utility functions are defined by
  \[
    u_2(a_2,a_3) = -u_3(a_2,a_3) = \begin{cases}\hphantom{-}1 & \text{if } a_2=a_3\\-1 & \text{if } a_2 \equiv a_3+1 \pmod k\\\hphantom{-}0 & \text{otherwise}
    \end{cases}
  \]
\end{definition}
We omit the easy analysis of the game $\calH_2(k)$.
\begin{lemma}
  For any $k\geq 2$, in the game $\calH_2(k)$ the strategy profile
  in which each action is played with probability $\tfrac{1}{k}$ is the
  unique Nash equilibrium and yields payoff~$0$ to both players.
  \label{LEM:Hbot}
\end{lemma}
\begin{definition}[The 3-player zero-sum game $\calG_2$]
  Let $\calG_2=\calG_2(\calS)$ be the game obtained from $\calG_1$ as
  follows. The action $\bot$ of Player~2 and Player~3 are replaced by
  the set of actions $(\bot,i)$, $i\in \{0,1,\dots,k-1\}$, where $k$
  is the maximum number of actions of a player in $\calG_1$. The
  payoff vector of the pure strategy profile
  $(\bot,(\bot,a_2),(\bot,a_3))$ is
  $(-2,1+u_2(a_2,a_3),1+u_3(a_2,a_3))$, where $u_2$ and $u_3$ are the
  utility functions of the game $\calH_2(k)$. Otherwise, when at least
  one player plays the action $G$, the payoff is as in $\calH_1$,
  where actions of the form $(\bot,i)$ are translated to the action
  $\bot$.
\end{definition}
\begin{theorem}
  \ExistsNEWithSmallSupports\ is $\cETR$-complete, even for 3-player
  zero-sum games.
  \label{THM:NEWithSmallSupports}
\end{theorem}
\begin{proof}
  In $\calG_2$, the strategy profile where Player~1 plays $\bot$ and
  Player~2 and Player~3 play $(\bot,i)$, with $i$ chosen uniformly at
  random, is a Nash equilibrium that takes the role of the Nash
  equilibrium $\allbot$ in $\calG_1$. Consider now an arbitrary Nash
  equilibrium in $\calG_2$. In case all players play the action $G$
  with probability less than~$1$, Player~2 and Player~3 must chose
  each action of the form $(\bot,i)$ with the same probability, since
  $\calH_2$ has a unique Nash equilibrium. The Nash equilibrium
  induces a strategy profile in $\calG_1$, letting Player~2 and
  Player~3 play the action~$\bot$ with the total probability each
  player placed on the actions~$(\bot,i)$.  By definition of
  $\calH_2(k)$ the payoff vector of $\allbot$ in $\calG_1$ differs by
  at most~1 in each entry from the payoff vectors of
  $(\bot,(\bot,a_2),(\bot,a_3))$. The proof of Lemma~\ref{LEM:H1NE}
  and Proposition~\ref{PROP:G1NE} still holds when changing the payoff
  vector of $\allbot$ by at most~1 in each coordinate.  The strategy
  profile induced in $\calG_1$ must therefore be a Nash equilibrium in
  $\calG_1$. We conclude that in a Nash equilibrium $x$ of $\calG_2$,
  either Player~2 and Player~3 use strategies with support of size $k$
  or $x$ is a Nash equilibrium of $\calG_0$, where every player uses a
  strategy of support size strictly less than~$k$ and where every
  player receives payoff~0. Proposition~\ref{PROP:G0NE} thus gives a
  reduction showing $\cETR$-hardness.
\end{proof}

\subsection{Pareto Optimal and Strong Nash Equilibria}
\label{SEC:ParetoStrongNE}
For showing $\cETR$-hardness for $\ExistsNonStrongNE$ we first analyze
the Strong Nash equilibria in the game $\calH_1$. 
\begin{lemma}
  For $u\geq 0$, the Nash equilibrium $\allbot$ of $\calH_1(u)$ is a
  strong Nash equilibrium. For $u=0$, the Nash equilibrium $\allG$ of
  $\calH_1(u)$ is not a strong Nash equilibrium.
  \label{LEM:H1StrongNE}
\end{lemma}
\begin{proof}
  Consider first $u=0$ and the Nash equilibrium $\allG$. This is not a
  strong Nash equilibrium, since for instance Player~1 and Player~2
  could both increase their payoff by playing the strategy profile
  $(\bot,\bot,G)$. Consider next $u\geq 0$ and the Nash equilibrium
  $\allbot$. Since $\calH_1$ is a zero-sum game it is sufficient to
  consider possible coalitions of two players. Player~2 and Player~3
  are already receiving the largest possible payoff given that
  Player~1 is playing the strategy~$\bot$, and hence they do not have
  a profitable deviation. Consider then, by symmetry, the coalition
  formed by Player~1 and Player~2, and let them play $G$ with
  probabilities $p_1$ and $p_2$. A simple calculation shows that to
  increase the payoff of Player~1 requires $p_1p_2+4p_2-2p_1>0$ and to
  increase the payoff of Player~2 requires $p_1p_2-4p_2+p_1>0$. Adding
  these gives $p_1(2p_2-1)>0$ which implies $p_2>\tfrac{1}{2}$. But
  then $p_1p_2-4p_2+p_1 < 0$. Thus $\allbot$ is a strong Nash equilibrium.  
\end{proof}
\pagebreak
\begin{theorem}
$\ExistsNonStrongNE$ is $\cETR$-complete, even for 3-player zero-sum games.
\end{theorem}
\begin{proof}
  Proposition~\ref{PROP:G0NE} and Proposition~\ref{PROP:G1NE} together
  give a reduction establishing $\cETR$-hardness, since by
  Lemma~\ref{LEM:H1StrongNE} the Nash equilibrium $\allbot$ is a
  strong Nash equilibrium, and a Nash equilibrium of $\calG_0$ where
  every player receives payoff~$0$ is not a strong Nash equilibrium.
\end{proof}

In a zero-sum game, every strategy profile is Pareto optimal. Thus for
showing $\cETR$-hardness of $\ExistsNonParetoOptimalNE$ we consider
non-zero-sum games.
\begin{definition}
  For $u\geq0$, let $\calH_3=\calH_3(u)$ be the 3-player game given by
  the following matrices, where Player~1 selects the matrix, Player~2
  selects the row, Player~3 selects the column.
  \begin{table}[H]
  \captionsetup[subtable]{labelformat=empty}
  \begin{subtable}[t]{0.49\textwidth}
  \centering
  \begin{tabular}{c|c|c|}
    \multicolumn{1}{c}{} & \multicolumn{1}{c}{$G$} & \multicolumn{1}{c}{$\bot$}\\\cline{2-3}
    $G$ & $(2u,-u,-u)$ & $(0,0,0)$\\\cline{2-3}
    $\bot$ & $(\phantom{2u}\mathllap{0},\phantom{-u}\mathllap{0},\phantom{-u}\mathllap{0})$ & $(1,1,1)$\\\cline{2-3}
  \end{tabular}
  \caption{$G$}
\end{subtable}
\begin{subtable}[t]{0.49\textwidth}
  \centering
  \begin{tabular}{c|c|c|}
    \multicolumn{1}{c}{} & \multicolumn{1}{c}{$G$} & \multicolumn{1}{c}{$\bot$}\\\cline{2-3}
    $G$ & $(0,0,0)$ & $(1,1,1)$\\\cline{2-3}
    $\bot$ & $(1,1,1)$ & $(2,2,2)$\\\cline{2-3}
  \end{tabular}
  \caption{$\bot$}
\end{subtable}
\end{table}
\end{definition}
\begin{lemma}
  When $u>0$, the only Nash equilibrium of $\calH_3(u)$ is the pure
  strategy profile $\allbot$. When $u=0$ the only Nash equilibria of
  $\calH_3(u)$ are the pure strategy profiles $\allG$ and
  $\allbot$. For $u\geq 0$, $\allbot$ is Pareto optimal. For $u=0$,
   $\allG$ is not Pareto optimal.
  \label{LEM:H3ParetoNE}
\end{lemma}
\begin{proof}
  When $u=0$, clearly $\allG$ is a Nash equilibrium, which is Pareto
  dominated by $\allbot$. Likewise, clearly $\allbot$ is always a
  Pareto optimal Nash equilibrium. When $u>0$, the action $G$ is
  strictly dominated by the action~$\bot$ for Player~2 and Player~3,
  and hence they play~$\bot$ with probability~$1$ in a Nash
  equilibrium. The only best reply of Player~1 is to play~$\bot$ with
  probability~$1$ as well.
\end{proof}
Analogously to Definition~\ref{DEF:G1} we define the game
$\calG_3=\calG_3(\calS)$ to be the game extending $\calG_0$ with
$\calH_3$ replacing the role of $\calH_1$ and analogously to
Proposition~\ref{PROP:G1NE} any Nash equilibrium in $\calG_3$
different from $\allbot$, which is Pareto optimal, must by
Lemma~\ref{LEM:H3ParetoNE} be a Nash equilibrium of $\calG_0$ with
payoff profile $(0,0,0)$, which is not Pareto optimal. This gives the
$\cETR$-hardness part of the following theorem.
\begin{theorem}[Bilò and Mavronicolas~\cite{STACS:BiloM16}]
  $\ExistsNonParetoOptimalNE$ is $\cETR$- complete, even for 3-player
  games.
  \label{THM:NonParetoOptimalNE}
\end{theorem}

We next consider the problems $\ExistsStrongNE$ and
$\ExistsParetoOptimalNE$. We first outline a proof of membership in
$\cETR$, building on ideas of Gatti~et~al~\cite{AAMAS:GattiRS13} and
Hansen, Hansen, Miltersen, and
Sørensen~\cite{WINE:HansenHMS08}. Gatti~et~al.~proved that deciding
whether a given strategy profile $x$ of an $m$-player game $\calG$ is
a strong Nash equilibrium can be done in polynomial time. The crucial
insight behind this result that the question of whether a coalition of
$k\leq m$ players may all improve their payoff by together changing
their strategies can be recast into a question in a derived game about
the minmax value of an additional \emph{fictitious} player that has
only $k$ strategies. Hansen~et~al.~proved that in such a game, the
minmax value may be achieved by strategies of the other players that
are of support at most $k$.
\begin{lemma}[Hansen~et~al.~\cite{WINE:HansenHMS08}]
  Let $\calG$ be a $m+1$ player game and let $k=\Abs{S_{m+1}}$.  If
  there exists a strategy profile $x$ of the first $m$~players such
  that $u_{m+1}(x;a) \leq 0$ for all $a \in S_{m+1}$ then there also
  exists a strategy profile $x'$ of the first~$m$ players in which each
  strategy has support size at most $k$ and $u_{m+1}(x';a) \leq 0$ for
  all $a \in S_{m+1}$.
\label{LEM:SmallSupportMinmax}
\end{lemma}
We next give a generalization of the auxiliary game construction of
Gatti~et~al.\ that also allows us to treat Pareto optimal Nash
equilibria at the same time.
\begin{definition}[cf.\ Gatti~et~al~\cite{AAMAS:GattiRS13}]
  Let $\calG$ be an $m$-player game with strategy sets $S_i$ and
  utility functions $u_i$. Let $x$ be a strategy profile of $\calG$
  and let $B_1 \dot\cup B_2 \dot\cup B_3 = [m]$ be a partition of the
  players, let $k_i=\abs{B_i}$ and $k=k_1+k_2$. For $\eps>0$ consider
  the $(m+1)$-player auxiliary game
  $\calG'=\calG'_{x,\eps,(B_1,B_2,B_3)}$ defined as follows. For
  $i \in B_1 \cup B_2$ the strategy set of Player~$i$ is
  $S'_i=S_i$. For $i \in B_3$ the strategy set of Player~$i$ is
  $S_i=\{\bot\}$. Finally, the strategy set of Player~$m+1$ is
  $B_1 \cup B_2$. The utility function of Player~$m+1$ is defined as
  as follows. Let $a=(a'_1,\dots,a'_m,j)$ be a pure strategy profile
  of $\calG'$. Define the strategy profile $x^a$
  of $\calG$ letting $x^a_i=a_i$ for $i \in B_1 \cup B_2$ and
  $x^a_i=x_i$ for $i \in B_3$. We then let
  $u'_{m+1}(a)=u_j(x)-u_j(x^a)+\eps$ for $j\in B_1$ and
  $u'_{m+1}(a)=u_j(x)-u_j(x^a)$ for $j \in B_2$.
\end{definition}
The following is immediate from the definition of $\calG'$.
\begin{lemma}
  There exist a strategy profile $x'$ in $\calG$ that satisfies
  $u_i(x')>u_i(x)$ when $i \in B_1$, $u_i(x') \geq u_i(x)$ when
  $i \in B_2$, and $x'_i=x_i$ when $i \in B_3$ if and only if there
  exist $\eps>0$ and a strategy $x'$ in
  $\calG'_{x,\eps,(B_1,B_2,B_3)}$ of the first $m$~players such that
  $u'_{m+1}(x',j)\leq 0$ for all $j \in B_1 \cup B_2$.
  \label{LEM:MinMaxFormulation}
\end{lemma}
The task of deciding if a strategy $x$ is Pareto optimal amounts to
checking the condition of Lemma~\ref{LEM:MinMaxFormulation} for
$B_1=\{i\}$ and $B_2=[m]\setminus\{i\}$ for all~$i$ and to decide
whether $x$ is a strong Nash equilibrium amounts to checking the
condition for all nonempty $B_1\subseteq [m]$ while letting
$B_2=\emptyset$.

According to Lemma~\ref{LEM:SmallSupportMinmax} we may restrict our
attention to strategies $x'$ in $\calG'$ of supports of size at
most~$m$. Fixing such a set of supports $T_i \subseteq S_i$ for
$i \in B_1 \cup B_2$, we may formulate the question of existence of a
strategy $x'$, with $\support(x'_i) \subseteq T_i$ for
$i \in B_1 \cup B_2$ that satisfies the conditions of
Lemma~\ref{LEM:MinMaxFormulation} as an existentially quantified
first-order formula over the reals. For a fixed $x$ we need only
$1+m^2$ existentially quantified variables to describe $\eps$ and the
strategy $x'$. Since this is a constant number of variables, when as in
our case $m$ is a constant, the general decision procedure of Basu,
Pollack, and Roy~\cite{BasuPollackRoy08} runs in polynomial time in
the bitsize of coefficients, number of polynomials, and their
degrees, resulting in an overall polynomial
time algorithm. Now, adding a step of simply enumerating over all
nonempty $B_1 \subseteq [m]$ and all support sets of size~$m$ we
obtain the result of Gatti~el~al.\ that deciding whether a given
strategy profile $x$ is a strong Nash equilibrium can be done in
polynomial time. The same holds in a similar way for checking that a
strategy profile is a Pareto optimal Nash equilibrium.

In our case, when proving $\cETR$ membership the only input is the
game $\calG$, whereas the strategy profile $x$ will be given by a
block of existentially quantified variables. We then need to show how to
express that $x$ is a Pareto optimal or a strong Nash equilibrium by a
quantifier free formula over the reals with free variables $x$. This
will be possible by the fact that quantifier elimination, rather than
just decision, is possible for the first order theory of the
reals. The quantifier elimination procedure of
Basu~et~al.~\cite{BasuPollackRoy08} runs in time exponential in the
number of free variables, so we cannot apply it directly.

Instead we express the condition of Lemma~\ref{LEM:MinMaxFormulation}
for a strategy profile $x'$ that is constrained by
$\support(x'_i) \subseteq T_i$ for $i \in B_1 \cup B_2$ in terms of
additional free variables $\widetilde{u}'$ that take the place of the
values of the utility function $u'$ of $\calG'$. Since the supports of
$x'$ are restricted to size $m$, just $m^{m+1}$ variables are needed
to represent the utility to Player~$m+1$ on every such pure strategy
profile. For constant $m$, this is a constant number of variables, and
thus the quantifier elimination procedure of Basu~et~al.\ runs in
polynomial time and outputs a quantifier free formula over the reals
with free variables $\widetilde{u}'$ that expresses the condition of
Lemma~\ref{LEM:MinMaxFormulation} when the utilities $u'$ are given by
$\widetilde{u}'$. After this we substitute expressions for the
utilities $u'$ in terms of the variables $x$ for the
variables~$\widetilde{u}'$. The final formula is obtained, in an
analogous way to the decision question, by enumerating over the
appropriate sets $B_1$ and $B_2$ as well as all possible supports
$T_i$, obtaining a formula for each such choice and combining them to
a single formula with free variables $x$ expressing either that $x$ is
Pareto optimal or that $x$ is a strong Nash equilibrium. To the former
we add the simple conditions of $x$ being a Nash equilibrium. Finally
we existentially quantify over $x$ and obtain a formula expressing
either that $\calG$ has a Pareto optimal Nash equilibrium or that
$\calG$ has a strong Nash equilibrium. Since this formula was computed
in polynomial time given $\calG$ we obtain the following result.
\begin{proposition}
  $\ExistsStrongNE$ and $\ExistsParetoOptimalNE$ both belong to
  $\cETR$.
\label{PROP:StrongParetoETRMembership}
\end{proposition}

For showing $\cETR$-hardness we construct a new extension of $\calG_0$.
\begin{definition}
  For $u\geq0$, let $\calH_4=\calH_4(u)$ be the 3-player game given by
  the following matrices, where Player~1 selects the matrix, Player~2
  selects the row, Player~3 selects the column.
  \begin{table}[H]
  \captionsetup[subtable]{labelformat=empty}
  \begin{subtable}[t]{0.49\textwidth}
  \centering
  \begin{tabular}{c|c|c|}
    \multicolumn{1}{c}{} & \multicolumn{1}{c}{$G$} & \multicolumn{1}{c}{$\bot$}\\\cline{2-3}
    $G$ & $(\phantom{-3}\mathllap{2u},-u,-u)$ & $(-3,-3,\phantom{-}0)$\\\cline{2-3}
    $\bot$ & $(-3,\phantom{-u}\mathllap{0},\phantom{-u}\mathllap{-3})$ & $(-2,-2,-2)$\\\cline{2-3}
  \end{tabular}
  \caption{$G$}
\end{subtable}
\begin{subtable}[t]{0.49\textwidth}
  \centering
  \begin{tabular}{c|c|c|}
    \multicolumn{1}{c}{} & \multicolumn{1}{c}{$G$} & \multicolumn{1}{c}{$\bot$}\\\cline{2-3}
    $G$ & $(\phantom{-}0,-3,-3)$ & $(-2,-2,-2)$\\\cline{2-3}
    $\bot$ & $(-2,-2,-2)$ & $(-1,-1,-1)$\\\cline{2-3}
  \end{tabular}
  \caption{$\bot$}
\end{subtable}
\end{table}
\end{definition}
\begin{lemma}
  When $u>0$, the only Nash equilibrium of $\calH_4(u)$ is the pure
  strategy profile $\allbot$. When $u=0$, the only Nash equilibria of
  $\calH_4(u)$ are the pure strategy profiles $\allG$ and $\allbot$.
  Furthermore, when $u=0$, the Nash equilibrium $\allG$ is both a Pareto
  optimal and a strong Nash equilibrium.
  \label{LEM:H4NE}
\end{lemma}
\begin{proof}
  When $u=0$, clearly $\allG$ is a Nash equilibrium, which is both
  Pareto optimal and a strong Nash equilibrium. Likewise, clearly
  $\allbot$ is always a Nash equilibrium. When $u>0$, the action $G$
  is strictly dominated by the action~$\bot$ for Player~2 and
  Player~3, and hence they play~$\bot$ with probability~$1$ in a Nash
  equilibrium. The only best reply of Player~1 is to play~$\bot$ with
  probability~$1$ as well.
\end{proof}
Analogously to Definition~\ref{DEF:G1} we define the game
$\calG_4=\calG_4(\calS)$ to be the game extending $\calG_0$ with
$\calH_4$ replacing the role of $\calH_1$. We next establish $\cETR$-hardness
\begin{theorem}
  $\ExistsParetoOptimalNE$ and $\ExistsStrongNE$ are $\cETR$-complete,
  even for 3-player games.
  \label{THM:ParetoStrongNE}
\end{theorem}
\begin{proof}
  In $\calG_4$, the strategy profile $\allbot$, with payoff profile
  $(-1,-1,-1)$, is a Nash equilibrium that is neither Pareto optimal
  or a strong Nash equilibrium, since by Lemma~\ref{LEM:G0payoff} a
  strategy profile in $\calG_0$ in which Player~1 plays an action
  according to the uniform distribution has payoff profile $(0,0,0)$.
  
  Similarly to the proof of Theorem~\ref{THM:G1NE}, any Nash
  equilibrium $x$ in $\calG_4$ different from $\allbot$ must by
  Lemma~\ref{LEM:H4NE} be a Nash equilibrium of $\calG_0$ with payoff
  profile $(0,0,0)$. Since $\calG_0$ is a zero-sum game, any strategy
  that is Pareto dominating $x$ must involve the strategy~$\bot$ and
  is thus ruled out by Lemma~\ref{LEM:H4NE}. Therefore $x$ is
  Pareto-optimal. Now, $x$ is not necessarily a strong Nash
  equilibrium, but by Lemma~\ref{LEM:G0payoff}, letting Player~1
  instead play an action of $\calG_0$ according to the uniform
  distribution is also a Nash equilibrium of $\calG_0$ with payoff
  profile $(0,0,0)$, that furthermore ensures that any strategy
  profile of Player~2 and Player~3 in $\calG_0$ does not improve their
  payoffs. Also, by Lemma~\ref{LEM:G0payoff}, no coalition involving
  Player~1 can improve their payoff without playing the action~$\bot$.
  No coalition can however improve their payoff by a strategy profile
  involving the action~$\bot$, since all such payoff profiles result
  in a player receiving negative payoff. Thus $x'$ is a strong Nash
  equilibrium.

  We conclude that Proposition~\ref{PROP:G0NE} gives a reduction
  showing $\cETR$-hardness of both $\ExistsParetoOptimalNE$ and
  $\ExistsStrongNE$, thereby together with
  Proposition~\ref{PROP:StrongParetoETRMembership} completing the
  proof.
\end{proof}

\subsection{Irrational and Rational Nash Equilibria}
\label{SEC:IrrationalRationNE}
Starting with a quadratic system in which every solution must involve
an irrational valued variable allows us to obtain $\cETR$-hardness for $\ExistsIrrationalNE$.
\begin{theorem}
  $\ExistsIrrationalNE$ is $\cETR$-hard, even for 3-player zero-sum games.
\label{THM:ExistsIrrationalNE}
\end{theorem}
\begin{proof}
  The proof of Proposition~\ref{PROP:QUAD-CornerSimplex} constructs a
  polynomial time computable function that takes a system $\calS$ of
  quadratic equations and produces at new system $\calS'$ of quadratic
  equations $\calS'$. From this construction it follows that there is
  an affine function $F$ given by a matrix and a vector with rational
  entries such that the set of solutions of $\calS'$ is the inverse
  image under $F$ of the set of solutions of $\calS$.  Adding to
  $\calS$ the equation $x^2-2=0$, where $x$ is a new variable, ensures
  that every solution of $\calS$ and hence $\calS'$ is not rational
  valued. This also holds for the homogeneous bilinear system of
  equations $\calS''$ obtained from $\calS'$ by
  Proposition~\ref{PROP:BilinearSystem}.  By
  Proposition~\ref{PROP:G0NE} any Nash equilibrium of
  $\calG_0(\calS'')$ with payoff profile $(0,0,0)$ is thereby not
  rational valued. We conclude that Proposition~\ref{PROP:G1NE} gives
  a reduction showing $\cETR$-hardness of $\ExistsIrrationalNE$, since
  the Nash equilibrium $\allbot$ of $\calG_1(\calS'')$ is  a
  rational valued strategy profile.
\end{proof}
While Theorem~\ref{THM:ExistsIrrationalNE} shows that deciding
whether a Nash equilibrium that is not rational valued exists is
$\cETR$-hard, we do not know whether the problem $\ExistsIrrationalNE$
is even decidable.

We next consider the question of deciding whether a given game has a
rational valued Nash equilibrium. This problem is naturally
expressible in the existential theory of the rationals $\ETQ$, which
is however not known to be decidable. It is natural to ask whether the
problem $\ExistsRationalNE$ is also $\cETQ$-hard. An obstacle for
such a result however, is that we do not know a bound on the magnitude
of coordinates of rational solutions to quadratic equations similar to
the case of real numbers. We can however start from a promise version
of $\QUAD_\QQ$ and construct a reduction to $\ExistsRationalNE$. We
sketch the construction below.
\begin{definition}
  Let $\QUAD_\QQ(\Ball(\allzeros,1))$ denote the promise problem
  given by  $\QUAD_\QQ$ together with the promise that if the
  given quadratic system has a solution over $\QQ$, then a solution
  over $\QQ$ exists in the unit ball $\Ball(\allzeros,1)$.
  \label{DEF:QBallPromiseProblem}
\end{definition}
A simple scaling and translation give a reduction from the promise
problem of Definition~\ref{DEF:QBallPromiseProblem} to the analogue
over $\QQ$ of the promise problem of
Proposition~\ref{PROP:QUAD-CornerSimplex} and then further to the
analogue over $\QQ$ of the promise problem of
Proposition~\ref{PROP:BilinearSystem}.  We shall then construct a
modification of $\calG_1$ in which the Nash equilibrium $\allbot$ is
replaced by an irrational valued Nash equilibrium. Several examples of
3-player games are known that are without rational valued Nash
equilibria. We give below a simple 3-player \emph{zero-sum} game with
a unique Nash equilibrium that is irrational valued.
\begin{definition}
  Let $\calH_5$ be the 3-player zero-sum game where each player has
  the action set $\{1,2\}$, and the payoff vectors are given by the
  following two matrices, where Player~1 selects the matrix, Player~2
  selects the row, Player~3 selects the column.
  \begin{table}[H]
    \captionsetup[subtable]{labelformat=empty}
  \begin{subtable}[t]{0.49\textwidth}
  \centering
  \begin{tabular}{c|c|c|}
    \multicolumn{1}{c}{} & \multicolumn{1}{c}{$1$} & \multicolumn{1}{c}{$2$}\\\cline{2-3}
    $1$ & $(-4,2,2)$ & $(-2,1,1)$\\\cline{2-3}
    $2$ & $(-2,1,1)$ & $(\phantom{-}0,0,0)$\\\cline{2-3}
  \end{tabular}
  \caption{$1$}
\end{subtable}
\begin{subtable}[t]{0.49\textwidth}
  \centering
  \begin{tabular}{c|c|c|}
    \multicolumn{1}{c}{} & \multicolumn{1}{c}{$1$} & \multicolumn{1}{c}{$2$}\\\cline{2-3}
    $1$ & $(\phantom{-}0,0,0)$ & $(-2,1,1)$\\\cline{2-3}
    $2$ & $(-2,1,1)$ & $(-6,3,3)$\\\cline{2-3}
  \end{tabular}
  \caption{$2$}
\end{subtable}
\end{table}
\end{definition}
We omit the straightforward but tedious analysis of the game $\calH_5$.
\begin{lemma}
  The unique Nash equilibrium of $\calH_5$ has Player~1 playing
  action~$1$ with probability~$1-1/\sqrt{6}$, and both Player~2 and
  Player~3 playing the action~$1$ with probability~$3-\sqrt{6}$. The
  Nash equilibrium payoff profile is 
  $(-4(3-\sqrt{6}),2(3-\sqrt{6}), 2(3-\sqrt{6}))$.
\end{lemma}
We can now provide our hardness statement for $\ExistsRationalNE$.
\begin{theorem}
  \begin{sloppypar}
  There is a polynomial time reduction from the pronmise problem 
  $\QUAD_\QQ(\Ball(\allzeros,1))$ to  $\ExistsRationalNE$, and the output
  of the reduction is a 3-player zero-sum game.
\end{sloppypar}
\label{THM:ExistsRationalNE}  
\end{theorem}
\begin{proof}
  Let $\calS$ be a system of quadratic equations in $n$ variables such
  that either $\calS$ has no solutions in $\QQ^n$ or has a solution in
  $\QQ^n \cap \Ball(\allzeros,1))$. As explained above we may in
  polynomial time transform $\calS$ into a system $\calS'$ of
  homogeneous bilinear polynomials in $2(n+1)$ variables such that
  $\calS$ has a solution in $\QQ^n$ if and only if $\calS'$ has a
  solution in
  $(\QQ^{n+1} \times \QQ^{n+1}) \cap
  (\Simplex^n\times\Simplex^n)$. Define the 3-player zero-sum game
  $\calG_5=\calG_5(\calS')$ to be the game obtained from
  $\calG_1(\calS')$ as follows, similarly to the definition
  of~$\calG_2$.

  The action~$\bot$ is for all players replaced by actions $(\bot,1)$
  and $(\bot,2)$. When the players choose the pure strategy profile
  $((\bot,a_1),(\bot,a_2),(\bot,a_3))$ Player~1 receive utility
  $-2+\tfrac{1}{6}u_1(a_1,a_2,a_3)$ and Player~2 and Player~3 both
  receive utility
  $1+\tfrac{1}{6}u_2(a_1,a_2,a_3)=1+\tfrac{1}{6}u_3(a_1,a_2,a_3)$,
  where $u_1$, $u_2$, and $u_3$ are the utility functions of the game
  $\calH_5$. Thus the payoff profile $(-2,1,1)$ of the strategy
  profile $\allbot$ is perturbed by the payoffs of the game $\calH_6$,
  scaled by $\tfrac{1}{6}$ in order to ensure that each entry is
  perturbed by at most~$1$. As in the proof of
  Theorem~\ref{THM:NEWithSmallSupports}, a Nash equilibrium $x$ is
  either a Nash equilibrium of $\calG_0$ in which every player
  receives payoff~$0$, or is such that the players choose the actions
  $(\bot,a)$ according to the unique Nash equilibrium of
  $\calH_5$. Since the latter is irrational valued we conclude that if
  $x$ is a rational valued Nash equilibrium then $x$ is a rational
  valued Nash equilibrium of $\calG_0(\calS')$ in which every player
  receives payoff~$0$, which by Proposition~\ref{PROP:G0NE} implies a
  rational valued solution to $\calS'$. Likewise a rational valued
  solution of $\calS'$ in $\Simplex^n\times\Simplex^n$ gives a
  rational valued Nash equilibrium of $\calG_1(\calS')$, thereby
  completing the proof.  \end{proof}

\section{Decision Problems about Nash Equilibria in Symmetric Games}
\label{SEC:DecisionProblemsSNE}

In this section we consider variations of all the decision problems
considered in Section~\ref{SEC:DecisionProblemsNE}, where the given
input is now a finite strategic form \emph{symmetric} game $\calD$,
where every player share the same set $S$ of pure strategies, together
with auxiliary input. As before, $u$ denotes a rational number, $k$ an
integer, whereas we now consider a single subset $T \subseteq S$ of
actions. The decision problems are described by stating the property
that a \emph{symmetric} Nash equilibrium $x$ whose existence is to be
determined should satisfy. We use the same grouping as the problems of
Section~\ref{SEC:DecisionProblemsNE}, but now we cover all problems in the
same section.
\begin{center}
\begin{tabular}{ll}\toprule
  Problem & Condition\\\midrule
\ExistsSNEWithLargePayoffs & $u_i(x)\geq u$ for all~$i$.\\
\ExistsSNEWithSmallPayoffs & $u_i(x)\leq u$ for all~$i$.\\
\ExistsSNEWithLargeTotalPayoff & $\sum_i u_i(x)\geq u$.\\
\ExistsSNEWithSmallTotalPayoff & $\sum_i u_i(x)\leq u$.\\\midrule
\ExistsSNEInABall & $x_i(a_i)\leq u$ for all~$i$ and $a_i\in S_i$.\\
\ExistsSecondSNE & $x$ is not the only SNE.\\
\ExistsSNEWithLargeSupports & $\abs{\support(x_i)}\geq k$ for all~$i$.\\
\ExistsSNEWithSmallSupports &  $\abs{\support(x_i)}\leq k$ for all~$i$.\\
\ExistsSNEWithRestrictingSupports &  $T \subseteq \support(x_i)$ for all~$i$.\\
  \ExistsSNEWithRestrictedSupports & $\support(x_i) \subseteq T$ for all~$i$.\\\midrule
\ExistsNonParetoOptimalSNE & $x$ is not Pareto optimal.\\
\ExistsNonStrongSNE & $x$ is not a strong NE.\\
\ExistsParetoOptimalSNE & $x$ is Pareto optimal.\\
\ExistsStrongSNE & $x$ is a strong NE.\\\midrule
\ExistsIrrationalSNE & $x_i(a_i) \not\in \QQ$ for some $i$ and $a_i\in S_i$.\\
  \ExistsRationalSNE & $x_i(a_i) \in \QQ$ for all $i$ and $a_i \in S_i$.\\\bottomrule  
\end{tabular}
\end{center}
In addition to the above problems about symmetric Nash equilibria, we
also shall consider the problem $\ExistsNonSymmetricNE$, that given a
finite strategic form symmetric game $\calD$, asks whether $\calD$ has
a Nash equilibrium $x$ that is \emph{nonsymmetric}.

$\cETR$ membership of all these problems, except for those of the last
group above, follows analogously to the case of their non-symmetric
counterparts and will not be discussed further.

\subsection{Symmetrization}

Garg~et~al.~\cite{TEAC:GargMVY18} constructed a \emph{symmetrization}
transformation of 3-player games to symmetric 3-player games. This was
used to give reductions from the two problems
$\ExistsNEWithRestrictingSupports$ and
$\ExistsNEWithRestrictedSupports$ to  their symmetric
counterparts, and these were the first problems about symmetric Nash
equilibria shown to be $\cETR$-complete.  Bilò and
Mavronicolas~\cite{STACS:BiloM17}, then constructed further reductions
starting from $\ExistsSNEWithRestrictedSupports$.

We can apply a different, but similar symmetrization transformation to
the game $\calG_0(\calS)$ of Section~\ref{SEC:DecisionProblemsNE}
obtaining a symmetric game $\calD_0(\calS)$ that will form the base of
further reduction as well as giving a direct proof of
$\cETR$-completeness for the problem $\ExistsSNEWithLargePayoffs$. In
addition to our new results, we give for completeness also proofs of
the previous $\cETR$-completeness results.

The idea of symmetrization is to take a game $\calG$, with strictly
positive payoffs, and construct a new symmetric game $\calD$ in which
the players can take the role of any player of $\calG$. The game
$\calG$ is then played when the players choose distinct roles. The
players are in the construction of
Garg~et~al.~\cite[Lemma~5.1]{TEAC:GargMVY18} incentivized to have this
behavior by the choice of payoffs ($0$ or $1$) in case the roles of
the players overlap. In our case we can simply let the players be
incentivized by the given payoff requirement alone.
\begin{definition}[The symmetric 3-player game $\calD_0$]
  Let $\calG_+=\calG_+(\calS)$ be the game obtained from
  $\calG_0(\calS)$ as follows. Let $u_1$, $u_2$, and $u_3$ be the
  utility functions of $\calG_0$. Let $M$ the the smallest (positive)
  integer such that $-M<u_1(x)<M$ for all pure strategy profiles
  $x$. Define the utility functions $u'_1$ and $u'_2=u'_3$ of
  $\calG_+$ by $u'_1(x)=u_1(x)+M$ and
  $u'_2(x)=u'_3(x)=-u_1(x)+M$. Thus also,
  $u'_2(x)=u'_3(x)=2u_2(x)+M=2u_3(x)+M$.

  \begin{sloppypar}
  For a permutation $\pi$ of $\{1,2,3\}$ we denote by
  $\calG_+^\pi=\calG_+^{(\pi(1),\pi(2),\pi(3))}$ the game where
  Player~$i$ has the set of actions $S_{\pi(i)}$ and the utility
  function  given by  $u'_{\pi(i)}(a_{\pi^{-1}(1)},a_{\pi^{-1}(2)},a_{\pi^{-1}(3)})$,
  where $a_i \in S_{\pi(i)}$ is the action chosen by Player~$i$. Thus
  $\calG_+^{\pi}$ is just a reordering of the players of $\calG_+$
  such that Player~$i$ in $\calG_+^{\pi}$ assumes the role of
  Player~$\pi(i)$ in $\calG_+$.
\end{sloppypar}

Define the game $\calD_0=\calD_0(\calS)$ to be the 3-player
  symmetric form game in which the players have the set of actions
  $S = S_1 \dot\cup S_2 \dot\cup S_3$, which is the \emph{disjoint}
  union of the set of actions $S_1$, $S_2$ and $S_3$ of the players in
  $\calG_0$. We also view $S_1$, $S_2$, and $S_3$ as disjoint sets
  below. When the players play actions $a_1$, $a_2$, and $a_3$, such
  that there exists a permutation $\pi$ of $\{1,2,3\}$ satisfying that
  $a_i \in S_{\pi(i)}$, for all~$i$, then Player~$i$ receives utility
  $u'_{\pi(i)}(a_{\pi^{-1}(1)},a_{\pi^{-1}(2)},a_{\pi^{-1}(3)})$. Otherwise,
  Player~$i$ simply receives utility~$0$. The payoffs vectors of
  $\calD_0$ are illustrated below as a block tensor of payoff vectors,
  where Player~1 selects the matrix slice, Player~2 selects the
  row, and Player~3 selects the column. We let $\mathbf{0}$ denote a
  payoff tensor of any appropriate dimensions in which every payoff
  is~$0$.
  \begin{table}[H]
    \captionsetup[subtable]{labelformat=empty}
\renewcommand{\arraystretch}{1.5}    
  \begin{subtable}[t]{0.49\textwidth}
  \centering
  \begin{tabular}{c|c|c|c|}
    \multicolumn{1}{c}{} & \multicolumn{1}{c}{$S_1$} & \multicolumn{1}{c}{$S_2$} & \multicolumn{1}{c}{$S_3$}\\\cline{2-4}
    $S_1$ & $\mathbf{0}$ & $\mathbf{0}$ & $\mathbf{0}$\\\cline{2-4}
    $S_2$ & $\mathbf{0}$ & $\mathbf{0}$ & $\calG_+^{(1,2,3)}$\\\cline{2-4}
    $S_3$ & $\mathbf{0}$ & $\calG_+^{(1,3,2)}$ & $\mathbf{0}$\\\cline{2-4}
  \end{tabular}
  \caption{$S_1$}
\end{subtable}
  \begin{subtable}[t]{0.49\textwidth}
  \centering
  \begin{tabular}{c|c|c|c|}
    \multicolumn{1}{c}{} & \multicolumn{1}{c}{$S_1$} & \multicolumn{1}{c}{$S_2$} & \multicolumn{1}{c}{$S_3$}\\\cline{2-4}
    $S_1$ & $\mathbf{0}$ & $\mathbf{0}$ & $\calG_+^{(2,1,3)}$\\\cline{2-4}
    $S_2$ & $\mathbf{0}$ & $\mathbf{0}$ & $\mathbf{0}$\\\cline{2-4}
    $S_3$ & $\calG_+^{(2,3,1)}$ & $\mathbf{0}$ & $\mathbf{0}$\\\cline{2-4}
  \end{tabular}
  \caption{$S_2$}
\end{subtable}
\\
\begin{subtable}[t]{\textwidth}
  \centering
  \begin{tabular}{c|c|c|c|}
    \multicolumn{1}{c}{} & \multicolumn{1}{c}{$S_1$} & \multicolumn{1}{c}{$S_2$} & \multicolumn{1}{c}{$S_3$}\\\cline{2-4}
    $S_1$ & $\mathbf{0}$ & $\calG_+^{(3,1,2)}$ & $\mathbf{0}$ \\\cline{2-4}
    $S_2$ & $\calG_+^{(3,2,1)}$ & $\mathbf{0}$ & $\mathbf{0}$\\\cline{2-4}
    $S_3$ & $\mathbf{0}$ & $\mathbf{0}$ & $\mathbf{0}$\\\cline{2-4}
  \end{tabular}
  \caption{$S_3$}
\end{subtable}
\end{table}
\label{DEF:D0}
\end{definition}
We next relate symmetric Nash equilibria in $\calD_0$ to Nash
equilibria in $\calG_0$.
\begin{lemma}
  The games $\calG_0(\calS)$ and $\calG_+(\calS)$ have the same set of
  Nash equilibria. All players receive payoff~$0$ in $\calG_0$ if and
  only if the total payoff of the players in $\calG_+$ is $3M$, which
  is also maximum possible total payoff of the players in $\calG_+$ in
  any Nash equilibrium.
\label{LEM:G+NE}
\end{lemma}
\begin{proof}
  Since the utility functions of $\calG_+$ are obtained from those of
  $\calG_0$ by scaling with a positive constant and adding a constant,
  the games have the same set of Nash equilibria. Note now that
  $u'_1(x)+u'_2(x)+u'_3(x)=3M-u_1(x)$. Since $u_1(x)\geq 0$ in any
  Nash equilibrium by Lemma~\ref{LEM:G0payoff}, the maximum total
  equilibrium payoff in $\calG_+$ is $3M$. In $\calG_0$ all players
  receive payoff~$0$ if and only if $u_1(x)=0$, from which the
  conclusion follows.
\end{proof}

\begin{proposition}
  Define $K=\tfrac{2M}{9}$, where $M$ is given in
  Definition~\ref{DEF:D0}.  Let $(x_1,x_2,x_3)$ be a Nash equilibrium
  of $\calG_0$ in which every player receive payoff~$0$. Then the
  strategy profile $(y,y,y)$ in which every player chooses
  $i \in \{1,2,3\}$, each with probability~$\tfrac{1}{3}$, and plays an
  action according to~$x_i$ is a symmetric Nash equilibrium in
  $\calD_0(\calS)$ in which all players receive payoff~$K$. Conversely
  let $(y,y,y)$ be a symmetric Nash equilibrium of $\calD_0$ in which
  every player receives payoff~$K$, which is also the maximum possible
  payoff of a symmetric Nash equilibrium of $\calD_0$. Then the total
  probability given to actions of each set~$S_i$ is
  exactly~$\frac{1}{3}$. Define $x_i$ to be the conditional
  probability distribution on $S_i$ obtained from $y$ given that an
  action of $S_i$ is played. Then $(x_1,x_2,x_3)$ is a Nash
  equilibrium of $\calG_0$ in which every player receives payoff~$0$.
  \label{PROP:D0SNE}
\end{proposition}
\begin{proof}
  First, let $x=(x_1,x_2,x_3)$ be a Nash equilibrium of $\calG_0$ in
  which every player receive payoff~$0$. By Lemma~\ref{LEM:G+NE} $x$
  is also a Nash equilibrium of $\calG_+$ given total payoff $3M$. Let
  $y$ be the strategy that selects each $i \in \{1,2,3\}$ with
  probability~$\tfrac{1}{3}$ and then chooses an action according to a
  $x_i$. Then $(y,y,y)$ must be a symmetric Nash equilibrium of
  $\calD_0$, since if a player could improve payoff by a change to a
  different strategy $y'$, there would also be a way for one of the
  players to improve the payoff in $\calG_+$. Each player takes part
  in playing $\calG_+$ a total of~6 times, each chosen with
  probability~$\tfrac{1}{27}$, and taking the role of each player
  2~times. The payoff to each player is therefore equal to
  $\tfrac{2}{27}3M=K$ by Lemma~\ref{LEM:G+NE}.

  Assume now that $(y,y,y)$ is a symmetric Nash equilibrium of
  $\calD_0$. Let $p_i$ be the total probability given to actions of
  $S_i$, for $i\in\{1,2,3\}$. Clearly, if $p_i=0$ for some $i$ the
  players receive payoff~$0$ due to the symmetrization construction of
  $\calD_0$. Assume now that $p_i>0$ for all $i$. The conditional
  probability distributions $x_i$, obtained from $y$ given that an
  action of $S_i$ is played, are therefore well defined. The strategy
  profile $x=(x_1,x_2,x_3)$ is a Nash equilibrium of $\calG_+$ as
  otherwise a player could improve the payoff in $\calD_0$ as well. By
  Lemma~\ref{LEM:G+NE} the total payoff $U$ to the players in
  $\calG_+$ is at most $3M$, and equals $3M$ exactly when $x$ gives
  payoff~$0$ to all players in $\calG_0$. The total payoff to the
  players in $\calD_0$ is therefore equal to
  $6p_1p_2p_3U \leq p_1p_2p_318M$. By the AM-GM inequality
  $p_1p_2p_3 \leq (\tfrac{1}{3}(p_1+p_2+p_3))^3 = \tfrac{1}{27}$ with
  equality if and only if $p_1=p_2=p_3=\tfrac{1}{3}$. The maximum
  total payoff of the players is thus $\tfrac{2}{3}M=3K$, and
  obtaining this requires both that $p_1=p_2=p_3=\tfrac{1}{3}$ and
  $U=M$. Thus by Lemma~\ref{LEM:G+NE}, if $(y,y,y)$ give all players
  payoff~$K$ in $\calD_0$ then $(x_1,x_2,x_3)$ give all players
  payoff~$0$ in~$\calG_0$.
\end{proof}
\subsection{Decision Problems for Symmetric Nash Equilibria}

From Proposition~\ref{PROP:D0SNE} together with
Theorem~\ref{THM:NEWithLargePayoffs} we immediately obtain the first
$\cETR$-hardness result about symmetric Nash equilibria.
\begin{theorem}[Bilò and Mavronicolas~\cite{STACS:BiloM17}]
  $\ExistsSNEWithLargePayoffs$ and \linebreak
  $\ExistsSNEWithLargeTotalPayoff$ are $\cETR$-complete, even for
  3-player games.
\end{theorem}
As done for the game $\calG_0$ we now construct simple extensions of
the game $\calD_0$. We describe these constructions below. For some of
the results we give only a proof sketch.
\begin{definition}[The symmetric 3-player game $\calD_1$]
  Let $\calD_1=\calD_1(\calS)$ be the game obtained from
  $\calD_0(\calS)$ as follows. Each player is given an additional
  action~$\bot$. When no player plays the action $\bot$, the payoffs
  are the same as in $\calD_0$. When exactly one player is
  playing~$\bot$, every player receives payoff $K$. When more than one
  player is playing~$\bot$, every player receives payoff $K+1$.
  \label{DEF:D1}
\end{definition}
\begin{proposition}
  The pure strategy profile $\allbot$ is a symmetric Nash equilibrium
  of $\calD_1$ in which every player receives payoff~$K+1$. Any other
  symmetric Nash equilibrium is also a symmetric Nash equilibrium of
  $\calD_0$ and is such that every player receives payoff~$K$.
  \label{PROP:D1NE}
\end{proposition}
\begin{proof}
  \begin{sloppypar}
  Let $(y,y,y)$ be a symmetric Nash equilibrium of $\calD_1$ that is
  different from $\allbot$. Let $y'$ be the probability distribution
  obtained from $y$ given that $\bot$ is not played. Then $(y',y',y')$
  must be a symmetric Nash equilibrium of $\calD_0$ in which every
  player receive payoff~$K$, since otherwise a player could improve
  the payoff in $\calD_1$ by always playing~$\bot$. Also it follows
  that $\bot$ is actually played with probability~$0$ by $y$, since
  otherwise a player could improve the payoff in $\calD_1$ by always
  playing~$\bot$. Thus $y=y'$ and the result follows. \qedhere
\end{sloppypar}
\end{proof}

The game $\calD_1$ gives, together with Proposition~\ref{PROP:G0NE},
reductions from the promise problem of
Proposition~\ref{PROP:BilinearSystem} to most of the problems under
consideration, showing $\cETR$-completeness. Except for
$\ExistsSNEInABall$, this was shown earlier by
Garg~et~al~\cite{TEAC:GargMVY18} and Bilò and
Mavronicolas~\cite{STACS:BiloM17}.
\begin{theorem}[Garg~et~al~\cite{TEAC:GargMVY18}; Bilò and Mavronicolas~\cite{STACS:BiloM17}]
  The following problems are $\cETR$-complete, even for 3-player games:
\begin{center}
  \begin{tabular}{ll}
$\ExistsSNEWithSmallPayoffs$, &
$\ExistsSNEWithSmallTotalPayoff$, \\
$\ExistsSNEInABall$, &
$\ExistsSecondSNE$, \\
$\ExistsSNEWithLargeSupports$, &
$\ExistsSNEWithRestrictingSupports$, \\
$\ExistsNonParetoOptimalSNE$, &
$\ExistsSNEWithRestrictedSupports$, \\
$\ExistsNonStrongSNE$.
  \end{tabular}
\end{center}
\end{theorem}
\begin{proof}
  Proposition~\ref{PROP:G0NE}, Proposition~\ref{PROP:D0SNE}, and
  Proposition~\ref{PROP:D1NE} together give a reduction from the
  promise problem of Proposition~\ref{PROP:BilinearSystem} to all the
  problems under consideration thereby showing $\cETR$-hardness, when
  setting the additional parameters as follows. We let $u=K$ for
  $\ExistsSNEWithSmallPayoffs$ and we let $u=3K$ for the similar problem 
  \linebreak$\ExistsSNEWithSmallTotalPayoff$. For
  $\ExistsSNEInABall$ we let $u=\tfrac{1}{2}$ and for
  \linebreak$\ExistsSNEWithLargeSupports$ we let $k=2$. We let $T$ be the set of
  all actions except~$i$ for $\ExistsSNEWithRestrictingSupports$ and
  $\ExistsSNEWithRestrictedSupports$. \linebreak \qedhere \end{proof} We can proceed
in a similar way as Section~\ref{SEC:DecisionProblemsNE} for the
remaining problems concerning symmetric Nash equilibria. In order to
adapt the proof of Theorem~\ref{THM:NEWithSmallSupports}, we need to
replace the Nash equilibrium $\allbot$ in $\calD_1$ by a symmetric
Nash equilibrium with large supports. Bilò and
Mavronicolas~\cite[Lemma~4]{STACS:BiloM17} construct for any $k$ a
symmetric $m$-player zero-sum game with a unique symmetric Nash
equilibrium that is fully mixed on a set of $k$ strategies. We may use
this to perturb the payoff profile of $\allbot$ in $\calD_1$
analogously to the proof of Theorem~\ref{THM:NEWithSmallSupports}
thereby obtaining an alternative proof of $\cETR$-hardness of
$\ExistsSNEWithSmallSupports$.
\pagebreak
\begin{theorem}[Bilò and Mavronicolas~\cite{STACS:BiloM17}]
  \begin{sloppypar}
  $\ExistsSNEWithSmallSupports$ is $\cETR$-complete, even for 3-player
  games. \qedhere
\end{sloppypar}
\end{theorem}
For the problems $\ExistsParetoOptimalSNE$ and $\ExistsStrongSNE$ we
define the game $\calD_4=\calD_4(\calS)$ extending $\calD_0$ in an
analogous way to the game $\calG_4$. Namely, each player is given an
additional action~$\bot$. When no player plays the action~$\bot$, the
payoffs are the same as in $\calD_0$. When exactly one player is
playing~$\bot$, that player receives payoff~$K$, whereas the other two
players receive payoff~$K-3$. When exactly two players are
playing~$\bot$, every player receives payoff~$K-2$. Finally, when all
players are playing~$\bot$, every player receives payoff~$K-1$. Thus
the utilities of the players when a player is playing the
action~$\bot$ are those of $\calH_4$ added to~$K$. In an analogous way
to the proof of Theorem~\ref{THM:ParetoStrongNE} we may then obtain
the following result.
\begin{theorem}
  $\ExistsParetoOptimalSNE$ and $\ExistsStrongSNE$ are
  $\cETR$-complete, even for 3-player games.
\end{theorem}
We now turn to irrational and rational valued symmetric Nash
equilibria. Analogously to the proof of
Theorem~\ref{THM:ExistsIrrationalNE}, starting with a quadratic system
$\calS$ in which every solution must involve an irrational valued
variable gives via the game $\calD_1$ a reduction showing
$\cETR$-hardness for $\ExistsIrrationalSNE$.
\begin{theorem}
  $\ExistsIrrationalSNE$ is $\cETR$-hard, even for 3-player games.
\end{theorem}
To make a symmetric analogue of Theorem~\ref{THM:ExistsRationalNE} we
need a 3-player symmetric game with unique Nash equilibrium that is
irrational valued. Rather than giving an explicit example, we note that
the symmetrization transformation of Garg~et~al.~\cite{TEAC:GargMVY18}
applied to, say, the game $\calH_5$ gives precisely such a symmetric
game. Using that to extend $\calD_1$ and perturb the payoff profile of
$\allbot$ we may obtain the following hardness result.
\begin{theorem}
  \begin{sloppypar}
  There is a polynomial time reduction from the promise problem 
  $\QUAD_\QQ(\Ball(\allzeros,1))$ to  $\ExistsRationalNE$.
\qedhere\end{sloppypar}
\end{theorem}

\subsection{A Decision Problem about Nonsymmetric Equilibria}
Our final result is concerned with the existence of a non-symmetric
Nash equilibrium in a symmetric game. Our hardness proof is based by a
modification of the games $\calD_0$ and $\calD_1$. We note that the
game $\calD_0$ was defined to be a symmetrization of the game
$\calG'_0$, used in the $\cETR$-hardness proof of Theorem~2 of for the
problem $\ExistsNEWithLargeTotalPayoff$, with $M$ added to every
payoff in order to make all payoffs strictly positive. This is the
appropriate choice for studying symmetric Nash equilibria, since in a
symmetric Nash equilibria of $\calD_0$ each player takes the role of
every player of $\calG_0$, thereby accumulating the payoffs of each
player (scaled appropriately). For studying nonsymmetric Nash
equilibria the idea is force the players to take on the role of just
one player of $\calG_0$.

Define $\calG'_+=\calG'_+(\calS)$ to be the game obtained from
$\calG_0$ by adding $M$ to all payoffs, where $M$ is the smallest
positive integer such that $-M < u_1(x) < M$. Define $\calD'_0$
analogously to $\calD_0$ with the game $\calG'_+$ taking the role of
$\calG_+$. Next, define the game $\calD'_1=\calD'_1(\calS)$ obtained
from $\calD'_0(\calS)$ by giving each player an additional
action~$\bot$, and defining the utility function as follows. When no
player plays the action~$\bot$, the payoffs are the same as in
$\calD'_0$. When exactly one player is playing~$\bot$, every player
receives payoff~$M$. When exactly two players are playing~$\bot$,
every player receives payoff~$M+1$. Finally, when all players are
playing~$\bot$, every player receives payoff~$M+2$.

\pagebreak 
\begin{theorem}
  $\ExistsNonSymmetricNE$ is $\cETR$-complete, even for 3-player games.
\end{theorem}
\begin{proof}
  We show $\cETR$-hardness by reduction from the promise problem of
  Proposition~\ref{PROP:BilinearSystem} by the game $\calD'_1(\calS)$.
  Consider a strategy profile $x'=(x_1,x_2,x_3)$ in the game
  $\calD'_0$. Since $\calG_0$ is a zero-sum game, the total payoff
  received by the players is at most $3M$. Furthermore, this is by the
  construction of $\calD'_0$ achievable only when there is a
  permutation $\pi$ of $\{1,2,3\}$ such that
  $\support(x'_i) \subseteq S_{\pi(i)}$, where $S_1$, $S_2$, and $S_3$ are
  the strategy sets of the players in $\calG_0$. Thus when the total
  payoff of the players is~$3M$ we may view the strategy profile $x'$
  as a strategy profile of $\calG_0(\calS)$.

  If there exists a strategy profile $x'$ in $\calG_0$ in which every
  player receives payoff~$0$, we may conversely view this as a
  (nonsymmetric) strategy profile of $\calD_0(\calS)$ in which every player
  receives payoff~$M$. This is also a Nash equilibrium in $\calD_1$
  which is nonsymmetric.
  
  Conversely, consider a Nash equilibrium $x$ of $\calD'_1$ that is
  nonsymmetric, and therefore different from $\allbot$. No player can
  play~$\bot$ with probability~$1$, since then $\bot$ would be the
  unique best reply of the other players. Thus we may consider the
  strategy profile $x'$ of $\calD'_0$ obtained from $x$ conditioned on
  that no player is playing~$\bot$. This must be a Nash equilibrium of
  $\calD'_0$ in which every player receives payoff~$M$, since
  otherwise $x$ would not be a Nash equilibrium of $\calD'_1$. As
  argued above this means that $x'$ gives a Nash equilibrium of
  $\calG_0$ in which every player receives payoff~$0$, thereby
  completing the proof using Proposition~\ref{PROP:G0NE}.
\end{proof}

\bibliographystyle{abbrv}
\bibliography{multiplayernash}

\begin{thebibliography}{10}

\bibitem{PJM:Aumann60}
R.~J. Aumann.
\newblock Acceptable points in games of perfect information.
\newblock {\em Pacific J. Math.}, 10(2):381--417, 1960.

\bibitem{BasuPollackRoy08}
S.~Basu, R.~Pollack, and M.-F. Roy.
\newblock {\em Algorithms in Real Algebraic Geometry}.
\newblock Springer, Berlin, Heidelberg, 2nd edition, 2008.

\bibitem{SAGT:BerthelsenH19}
M.~L.~T. Berthelsen and K.~A. Hansen.
\newblock On the computational complexity of decision problems about
  multi-player {Nash} equilibria.
\newblock In {\em {SAGT} 2019}, volume 11801 of {\em Lecture Notes in Computer
  Science}, pages 153--167. Springer, 2019.

\bibitem{TOCS:BiloM14}
V.~Bil{\`o} and M.~Mavronicolas.
\newblock Complexity of rational and irrational {Nash} equilibria.
\newblock {\em Theory of Computing Systems}, 54(3):491--527, 2014.

\bibitem{STACS:BiloM16}
V.~Bil{\`o} and M.~Mavronicolas.
\newblock A catalog of $\exists\mathbb{R}$-complete decision problems about
  {Nash} equilibria in multi-player games.
\newblock In N.~Ollinger and H.~Vollmer, editors, {\em {STACS} 2016}, volume~47
  of {\em LIPIcs}, pages 17:1--17:13. Schloss Dagstuhl - Leibniz-Zentrum
  f{\"u}r Informatik, 2016.

\bibitem{STACS:BiloM17}
V.~Bil\'o and M.~Mavronicolas.
\newblock $\exists\mathbb{R}$-complete decision problems about symmetric {Nash}
  equilibria in symmetric multi-player games.
\newblock In H.~Vollmer and B.~Vall\'e, editors, {\em {STACS} 2017}, volume~66
  of {\em LIPIcs}, pages 13:1--13:14. Schloss Dagstuhl--Leibniz-Zentrum f{\"u}r
  Informatik, 2017.

\bibitem{BAMS:BlumSS89}
L.~Blum, M.~Shub, and S.~Smale.
\newblock On a theory of computation and complexity over the real numbers:
  {NP}-completeness, recursive functions and universal machines.
\newblock {\em Bull. Amer. Math. Soc.}, 21(1):1--46, 1989.

\bibitem{FOCM:BurgisserC09}
P.~B{\"u}rgisser and F.~Cucker.
\newblock Exotic quantifiers, complexity classes, and complete problems.
\newblock {\em Foundations of Computational Mathematics}, 9(2):135--170, 2009.

\bibitem{JCSS:BussFS99}
J.~F. Buss, G.~S. Frandsen, and J.~O. Shallit.
\newblock The computational complexity of some problems of linear algebra.
\newblock {\em Journal of Computer and System Sciences}, 58(3):572 -- 596,
  1999.

\bibitem{STOC:Canny88}
J.~F. Canny.
\newblock Some algebraic and geometric computations in {PSPACE}.
\newblock In J.~Simon, editor, {\em Proceedings of the 20th Annual {ACM}
  Symposium on Theory of Computing (STOC 1988)}, pages 460--467. ACM, 1988.

\bibitem{FOCS:ChenDeng06}
X.~Chen and X.~Deng.
\newblock Settling the complexity of two-player {Nash} equilibrium.
\newblock In {\em 47th Annual {IEEE} Symposium on Foundations of Computer
  Science (FOCS 2006)}, pages 261--272. IEEE Computer Society Press, 2006.

\bibitem{GEB:ConitzerS08}
V.~Conitzer and T.~Sandholm.
\newblock New complexity results about {Nash} equilibria.
\newblock {\em Games and Economic Behavior}, 63(2):621--641, 2008.

\bibitem{SICOMP:DaskalakisGP09}
C.~Daskalakis, P.~W. Goldberg, and C.~H. Papadimitriou.
\newblock The complexity of computing a {Nash} equilibrium.
\newblock {\em SIAM J. Comput.}, 39(1):195--259, 2009.

\bibitem{SICOMP:EtessamiY10}
K.~Etessami and M.~Yannakakis.
\newblock On the complexity of {N}ash equilibria and other fixed points.
\newblock {\em SIAM J. Comput.}, 39(6):2531--2597, 2010.

\bibitem{TEAC:GargMVY18}
J.~Garg, R.~Mehta, V.~V. Vazirani, and S.~Yazdanbod.
\newblock $\exists\mathbb{R}$-completeness for decision versions of
  multi-player (symmetric) {Nash} equilibria.
\newblock {\em ACM Trans. Econ. Comput.}, 6(1):1:1--1:23, 2018.

\bibitem{AAMAS:GattiRS13}
N.~Gatti, M.~Rocco, and T.~Sandholm.
\newblock On the verification and computation of strong {Nash} equilibrium.
\newblock In M.~L. Gini, O.~Shehory, T.~Ito, and C.~M. Jonker, editors, {\em
  {AAMAS} 2013}, pages 723--730. IFAAMAS, 2013.

\bibitem{GEB:GilboaZemel89}
I.~Gilboa and E.~Zemel.
\newblock Nash and correlated equilibria: Some complexity considerations.
\newblock {\em Games and Economic Behavior}, 1(1):80--93, 1989.

\bibitem{TCS:Hansen19}
K.~A. Hansen.
\newblock The real computational complexity of minmax value and equilibrium
  refinements in multi-player games.
\newblock {\em Theory of Computing Systems}, 63(7), 2019.

\bibitem{WINE:HansenHMS08}
K.~A. Hansen, T.~D. Hansen, P.~B. Miltersen, and T.~B. S{\o}rensen.
\newblock Approximability and parameterized complexity of minmax values.
\newblock In C.~H. Papadimitriou and S.~Zhang, editors, {\em {WINE} 2008},
  volume 5385 of {\em Lecture Notes in Computer Science}, pages 684--695.
  Springer, 2008.

\bibitem{AnnMath:Koenigsmann16}
J.~Koenigsmann.
\newblock Defining $\mathbb{Z}$ in $\mathbb{Q}$.
\newblock {\em Annals of Mathematics}, 183(1):73--93, 2016.

\bibitem{SAGT:MehtaVY15}
R.~Mehta, V.~V. Vazirani, and S.~Yazdanbod.
\newblock Settling some open problems on 2-player symmetric {Nash} equilibria.
\newblock In M.~Hoefer, editor, {\em {SAGT} 2015}, volume 9347 of {\em Lecture
  Notes in Computer Science}, pages 272--284. Springer, 2015.

\bibitem{AM:Nash51}
J.~Nash.
\newblock Non-cooperative games.
\newblock {\em Annals of Mathematics}, 2(54):286--295, 1951.

\bibitem{JSC:Renegar92}
J.~Renegar.
\newblock On the computational complexity and geometry of the first-order
  theory of the reals, part {I-III}.
\newblock {\em J. Symb. Comput}, 13(3):255--352, 1992.

\bibitem{GD:Schaefer09}
M.~Schaefer.
\newblock Complexity of some geometric and topological problems.
\newblock In D.~Eppstein and E.~R. Gansner, editors, {\em {GD} 2009}, volume
  5849 of {\em LNCS}, pages 334--344. Springer, 2010.

\bibitem{TOCS:SchaeferS15}
M.~Schaefer and D.~\v{S}tefankovi\v{c}.
\newblock Fixed points, {Nash} equilibria, and the existential theory of the
  reals.
\newblock {\em Theory of Computing Systems}, 60:172--193, 2017.

\end{thebibliography}

\end{document}